\documentclass{article}

\usepackage{fullpage}

\usepackage[all=normal,bibliography=tight]{savetrees}

\usepackage{tikz}
\usetikzlibrary{snakes}
\usetikzlibrary{positioning}
\usetikzlibrary{decorations.markings}
\usepackage{xspace}
\usepackage{comment}

\usepackage{amsmath,amstext,amssymb,amsthm}

\newtheorem{theorem}{Theorem}[section]
\newtheorem{lemma}[theorem]{Lemma}
\newtheorem{definition}[theorem]{Definition}

\newtheorem{proposition}[theorem]{Proposition}

\theoremstyle{definition}
\newtheorem{example}[theorem]{Example}

\begin{document}

\newcommand{\schedname}{{\sc{SCHED}}\xspace}
\newcommand{\schedlongname}{$1|\textrm{prec}|\sum C_i$\xspace}
\newcommand{\czas}[1]{t(#1)}
\newcommand{\koszt}[1]{T(#1)}
\newcommand{\tudu}[1]{{\bf{TODO: #1}}}
\newcommand{\eps}{\varepsilon}
\newcommand{\pinezka}{\textrm{nil}}
\newcommand{\pred}{{\ensuremath{pred}}}
\newcommand{\succc}{{\ensuremath{succ}}}
\renewcommand{\subset}{\subseteq}
\newcommand{\Oh}{\ensuremath{O}}
\newcommand{\Ohstar}{\ensuremath{\Oh^\ast}}

\newcommand{\Wvc}{M}
\newcommand{\Whalf}{W_{\mathrm{half}}}
\newcommand{\Wquarter}{W_{\mathrm{quarter}}}
\newcommand{\matching}{\mathcal{M}}

\newcommand{\defproblemnoparam}[3]{
  \vspace{1mm}
\noindent\fbox{
  \begin{minipage}{\textwidth}  
  #1 \\ 
  {\bf{Input:}} #2  \\
  {\bf{Task:}} #3
  \end{minipage}
  }
  \vspace{1mm}
}

\definecolor{light-gray}{gray}{0.8}

\title{Scheduling partially ordered jobs faster than $2^n$\thanks{An extended abstract of this paper appears at 19th European Symposium on Algorithms, Saarbr\"{u}cken, Germany, 2011.}}

  \author{Marek Cygan\thanks{Institute of Informatics, University of Warsaw, Poland, \texttt{cygan@mimuw.edu.pl}. Supported by Polish Ministry of Science grant no. N206 355636 and Foundation for Polish Science.} \and
	Marcin Pilipczuk\thanks{Institute of Informatics, University of Warsaw, Poland, \texttt{malcin@mimuw.edu.pl}. Supported by Polish Ministry of Science grant no. N206 355636 and Foundation for Polish Science.} \and
	Micha\l{} Pilipczuk\thanks{Faculty of Mathematics, Informatics and Mechanics, University of Warsaw, Poland, \texttt{michal.pilipczuk@students.mimuw.edu.pl}} \and
	Jakub Onufry Wojtaszczyk\thanks{Google Inc., Cracow, Poland, \texttt{onufry@google.com}}}

\date{}

\maketitle

\begin{abstract}
In a scheduling problem, denoted by \schedlongname in the Graham notation,
we are given a set of $n$ jobs, together with their processing
times and precedence constraints. The task is to order the jobs so that their total completion time is minimized.
\schedlongname is a special case of the Traveling Repairman Problem with precedences.
A natural dynamic programming algorithm solves both these problems in $2^n n^{\Oh(1)}$ time, and whether
there exists an algorithms solving \schedlongname{} in $\Oh(c^n)$ time for some constant $c<2$
was an open problem posted in 2004 by Woeginger. In this paper we answer this question positively.
%\keywords{moderately exponential algorithms, jobs scheduling, jobs with precedences}
\end{abstract}

\section{Introduction}

It is commonly believed that no NP-hard problem is solvable in polynomial time.
However, while all NP-complete problems are equivalent with respect to polynomial time reductions,
they appear to be very different with respect to the best exponential time exact solutions.
In particular, most NP-complete problems can be solved significantly faster than
the (generic for the NP class) obvious brute-force algorithm that checks all possible solutions;
examples are {\sc Independent Set}~\cite{fgk:m-c-jacm}, {\sc Dominating Set}~\cite{fgk:m-c-jacm,rooij:domsetesa09}, {\sc Chromatic Number}~\cite{bjohus:color} and
{\sc Bandwidth}~\cite{nasz-tcs}.
The area of moderately exponential time algorithms studies upper and lower bounds
for exact solutions for hard problems.
The race for the fastest exact algorithm inspired several very interesting tools
and techniques such as Fast Subset Convolution~\cite{bjohus:fourier} and Measure\&Conquer~\cite{fgk:m-c-jacm}
(for an overview of the field we refer the reader to a recent book by Fomin and Kratsch~\cite{fomin-book}).

For several problems, including {\sc TSP}, {\sc Chromatic Number}, {\sc Permanent},
{\sc Set Cover}, {\sc \#Hamiltonian Cycles} and {\sc SAT}, the currently best known time complexity
is of the form\footnote{The $\Ohstar()$ notation suppresses factors polynomial in the input size.}
$\Ohstar(2^n)$, which is a result of applying dynamic programming over subsets,
the inclusion-exclusion principle or a brute force search.
The question remains, however, which of those problems are inherently so hard that it is
not possible to break the $2^n$ barrier and which are just waiting for new tools and techniques still to be discovered.
In particular, the hardness of the $k$-{\sc SAT} problem is the starting point for the
Strong Exponential Time Hypothesis of Impagliazzo and Paturi~\cite{seth},
which is used as an argument that other problems are hard~\cite{cut-and-count,treewidth-lower,patrascu}.
Recently, on the positive side, $\Oh(c^n)$ time algorithms for a constant $c<2$ have been developed for
{\sc Capacitated Domination}~\cite{capdomset}, {\sc Irredundance}~\cite{irrset},
{\sc Maximum Induced Planar Subgraph}~\cite{fomin-esa11} and
(a major breakthrough in the field)
for the undirected version of the {\sc Hamiltonian Cycle} problem~\cite{bjorklund-focs}.

In this paper we extend this list by one important scheduling problem.
The area of scheduling algorithms originates from practical questions regarding scheduling jobs on single- or multiple-processor
machines or scheduling I/O requests. It has quickly become one of the most important areas in algorithmics, with significant influence on other branches of computer science.
For example, the research of the job-shop scheduling problem in 1960s resulted
in designing the competitive analysis \cite{graham:jobshop}, initiating the research of
online algorithms.
Up to today, the scheduling literature consists of thousands of research publications.
We refer the reader to the classical textbook of Brucker \cite{sched:book}.

Among scheduling problems one may find a bunch of problems solvable in polynomial time, as well as many NP-hard ones.
For example, the aforementioned job-shop problem is NP-complete on at least three machines \cite{job-shop:3},
but polynomial on two machines with unitary processing times \cite{job-shop:2}.

Scheduling problems come in numerous variants. For example, one may consider
scheduling on one machine, or many uniform or non-uniform machines.
The jobs can have different attributes: they may arrive at different times,
may have deadlines or precedence constraints, preemption may or may not be allowed.
There are also many objective functions, for example the makespan of the computation,
total completion time, total lateness (in case of deadlines for jobs) etc.

Let us focus on the case of a single machine. Assume we are given a set of jobs $V$,
and each job $v$ has its processing time $t(v) \in [0,+\infty)$.
For a job $v$, its completion time is the total amount of time that this job waited
to be finished; formally, the completion time of a job $v$ is defined as the sum
of processing times of $v$ and all jobs scheduled earlier.
If we are to minimize the total completion time (i.e, the sum of completion times
over all jobs), it is clear that the jobs should be scheduled in order of increasing
processing times.
The question of minimizing the makespan of the computation (i.e., maximum completion time)
is obvious in this setting, but we note that minimizing makespan is polynomially solvable even
if we are given a precedence constraints on the jobs (i.e., a partial order on the set of jobs
is given, and a job cannot be scheduled before all its predecessors in the partial order
are finished) and jobs arrive at different times (i.e., each job has its arrival time,
  before which it cannot be scheduled) \cite{lawler:sched}.

Lenstra and Rinnooy Kan \cite{lenstra} in 1978 proved that
the question of minimizing total completion time on one machine becomes NP-complete
if we are given precedence constraints on the set of jobs.
To the best of our knowledge the currently smallest approximation ratio for this case equals $2$,
due to independently discovered algorithms by Chekuri and Motwani~\cite{chekuri} as well as Margot et al.~\cite{margot}.
The problem of minimizing total completion time on one machine, given
precedence constraints on the set of jobs, can be solved
by a standard dynamic programming algorithm in time $\Ohstar(2^n)$, where $n$ denotes the
number of jobs.
In this paper we break the $2^n$-barrier for this problem.

Before we start, let us define formally the considered problem.
As we focus on a single scheduling problem, for brevity we
denote it by \schedname{}. We note that the proper name of this problem in the Graham notation
    is \schedlongname{}.

\defproblemnoparam{\schedname{}}
{A partially ordered set of jobs $(V, \leq)$,
together with a nonnegative processing time $\czas{v} \in [0,+\infty)$ for each job $v \in V$.
}
{
Compute a bijection $\sigma:V \to \{1,2,\ldots,|V|\}$ (called an {\em{ordering}})
that satisfies the precedence constraints (i.e., if $u < v$, then $\sigma(u) < \sigma(v)$)
and minimizes the total completion time of all jobs defined as
$$\koszt{\sigma} = \sum_{v \in V}\ \ \sum_{u:\sigma(u) \leq \sigma(v)} \czas{u} = \sum_{v \in V} (|V|-\sigma(v)+1)\czas{v}.$$
}

If $u < v$ for $u,v \in V$ (i.e., $u \leq v$ and $u \neq v$),
we say that $u$ {\em{precedes}} $v$, $u$ is a {\em{predecessor}}
or {\em{prerequisite}} of $v$, $u$ is {\em{required}} for $v$ or that $v$ is a {\em{successor}}
of $u$. We denote $|V|$ by $n$.

\schedname{} is a special case of the precedence constrained Travelling Repairman Problem (prec-{\sc TRP}), defined as follows. A repairman needs to visit all vertices of a (directed
or undirected) graph $G=(V,E)$ with distances $d:E \to [0,\infty)$ on edges.
At each vertex, the repairman is supposed to repair a broken machine;
a cost of a machine $v$ is the time $C_v$ that it waited before being repaired.
Thus, the goal is to minimize the total repair time, that is, $\sum_{v \in V} C_v$.
Additionally, in the precedence constrained case, we are given a partial order $(V,\leq)$
on the set of vertices of $G$; a machine can be repaired only if all its predecessors
are already repaired.
Note that, given an instance $(V,\leq,t)$ of \schedname{},
we may construct equivalent prec-{\sc{TRP}} instance, by taking $G$ to be a complete directed
graph on the vertex set $V$, keeping the precedence constraints unmodified, and setting
$d(u,v) = t(v)$.

The {\sc{TRP}} problem is closely related to the Traveling Salesman Problem ({\sc{TSP}}).
All these problems are NP-complete and solvable in $\Ohstar(2^n)$ time by
an easy application of the dynamic programming approach (here $n$ stands for the number
of vertices in the input graph).
In 2010, Bj\"{o}rklund \cite{bjorklund-focs} discovered a genuine way to solve
probably the easiest NP-complete version of the {\sc{TSP}} problem --- the question
of deciding whether a given undirected graph is Hamiltonian --- in randomized
$\Oh(1.66^n)$ time. However, his approach does not extend to directed graphs,
not even mentioning graphs with distances defined on edges.

Bj\"{o}rklund's approach is based on purely graph-theoretical and combinatorial reasonings,
and seem unable to cope with arbitrary (large, real) weights (distances, processing times).
This is also the case with many other combinatorial approaches.
Probably motivated by this,
Woeginger at International Workshop on Parameterized and Exact Computation
(IWPEC) in 2004~\cite{woeginger04} has posed the question (repeated in 2008~\cite{woeginger08}),
whether it is possible to construct an $\Oh((2-\eps)^n)$ time algorithm for the \schedname{} problem\footnote{Although Woeginger in his papers asks for an $\Oh(1.99^n)$ algorithm, the intention is clearly to ask for an $\Oh((2-\eps)^n)$ algorithm.}.
This problem seems to be the easiest case of the aforementioned
family of {\sc{TSP}}-related problems with arbitrary weights.
In this paper we present such an algorithm, thus affirmatively answering Woeginger's question.
Woeginger also asked~\cite{woeginger04,woeginger08} whether an $\Oh((2-\eps)^n)$ time algorithm
for one of the problems {\sc TRP}, {\sc TSP}, prec-{\sc TRP}, \schedname{} implies $\Oh((2-\eps)^n)$
time algorithms for the other problems. This problem is still open.

The most important ingredient of our algorithm is a combinatorial lemma (Lemma~\ref{lem:core})
which allows us to investigate the structure of the \schedname{} problem.
We heavily use the fact that we are solving the \schedname{} problem and
not its more general TSP related version, and for this reason
we believe that obtaining $\Oh((2-\eps)^n)$ time algorithms for other problems
listed by Woeginger is much harder.

\section{The algorithm}

\subsection{High-level overview --- part 1}\label{sec:high-level1}
Let us recall that our task in the \schedname{} problem is to compute an ordering $\sigma:V \to \{1,2,\ldots,n\}$
that satisfies the precedence constraints (i.e., if $u < v$ then $\sigma(u) < \sigma(v)$)
and minimizes the total completion time of all jobs defined as
$$\koszt{\sigma} = \sum_{v \in V}\ \ \sum_{u:\sigma(u) \leq \sigma(v)} \czas{u} = \sum_{v \in V} (n-\sigma(v)+1)\czas{v}.$$
We define {\em{the cost of job $v$ at position $i$}} to be $\koszt{v, i} = (n-i+1)\czas{v}$.
Thus, the total completion time is the total cost of all jobs at their respective
positions in the ordering $\sigma$.

We begin by describing the algorithm that solves \schedname in $O^\star(2^n)$ time, which we call {\em{the DP algorithm}}
--- this will be the basis for our further work.
The idea --- a standard dynamic programming over subsets --- is that if we decide that a particular
set $X \subset V$ will (in some order) form the prefix of our optimal $\sigma$, then the order
in which we take the elements of $X$ does not affect the choices we make regarding the ordering
of the remaining $V\setminus X$; the only thing that matters are the precedence constraints
imposed by $X$ on $V\setminus X$. Thus, for each candidate set $X \subseteq V$ to form a prefix,
the algorithm computes a bijection
$\sigma[X]:X \to \{1,2,\ldots,|X|\}$ that minimizes the cost of jobs from $X$, i.e.,
it minimizes $\koszt{\sigma[X]} = \sum_{v \in X} \koszt{v, \sigma[X](v)}$.
The value of $\koszt{\sigma[X]}$ is computed using the following easy to check recursive formula:
\begin{equation}\label{eqn:reccost} \koszt{\sigma[X]} = \min_{v \in \max(X)} \left[\koszt{\sigma[X \setminus \{v\}]} + \koszt{v, |X|}\right].\end{equation}
Here, by $\max(X)$ we mean the set of maximum elements of $X$ --- those which do not precede any element of $X$.
The bijection $\sigma[X]$ is constructed by prolonging $\sigma[X\setminus\{v\}]$ by $v$,
where $v$ is the job at which the minimum is attained.
Notice that $\sigma[V]$ is exactly the ordering we are looking for.
We calculate $\sigma[V]$ recursively, using formula (\ref{eqn:reccost}),
storing all computed values $\sigma[X]$ in memory to avoid recomputation.
Thus, as the computation of a single $\sigma[X]$ value given all the smaller values
takes polynomial time, while $\sigma[X]$ for each $X$ is computed at most once
the whole algorithm indeed runs in $O^\star(2^n)$ time.

The overall idea of our algorithm is to identify a family of sets $X \subset V$ that --- for some
reason --- are not reasonable prefix candidates, and we can skip them in the computations of the
DP algorithm; we will call these {\em unfeasible sets}. If the number of feasible sets is not
larger than $c^n$ for some $c < 2$, we will be done --- our recursion will visit only
feasible sets, assuming $T(\sigma[X])$ to be $\infty$ for unfeasible $X$ in formula
(\ref{eqn:reccost}), and the running time will be $O^\star(c^n)$. This is formalized
in the following proposition.

\begin{proposition}\label{prop:cut-dp}
Assume we are given a polynomial-time algorithm $\mathcal{R}$ that, given a set $X \subseteq V$,
either accepts it or rejects it. Moreover, assume that the number of sets accepted by $\mathcal{R}$
is bounded by $\Oh(c^n)$ for some constant $c$. Then one can find in time $O^\star(c^n)$ an optimal
ordering of the jobs in $V$ among those orderings $\sigma$ where $\sigma^{-1}(\{1,2,\ldots,i\})$ is
accepted by $\mathcal{R}$ for all $1 \leq i \leq n$, whenever such ordering exists.
\end{proposition}
\begin{proof}
Consider the following recursive procedure to compute optimal $\koszt{\sigma[X]}$ for a given set $X \subseteq V$:
\begin{enumerate}
\item if $X$ is rejected by $\mathcal{R}$, return $\koszt{\sigma[X]} = \infty$;
\item if $X = \emptyset$, return $\koszt{\sigma[X]} = 0$;
\item if $\koszt{\sigma[X]}$ has been already computed, return the stored value of $\koszt{\sigma[X]}$;
\item otherwise, compute $\koszt{\sigma[X]}$ using formula \eqref{eqn:reccost}, calling recursively the procedure
itself to obtain values $\koszt{\sigma[X \setminus \{v\}]}$ for $v \in \max(X)$, and store the computed value
for further use.
\end{enumerate}
Clearly, the above procedure, invoked on $X = V$, computes optimal $\koszt{\sigma[V]}$ among those orderings
$\sigma$ where $\sigma^{-1}(\{1,2,\ldots,i\})$ is accepted by $\mathcal{R}$ for all $1 \leq i \leq n$.
It is straightforward to augment this procedure to return the ordering $\sigma$ itself, instead of only its cost.

If we use balanced search tree to store the computed values of $\sigma[X]$,
each recursive call of the described procedure runs in polynomial time.
Note that the last step of the procedure is invoked at most once for each set $X$ accepted by $\mathcal{R}$
and never for a set $X$ rejected by $\mathcal{R}$.
As an application of this step results in at most $|X| \leq n$ recursive calls,
we obtain that a computation of $\sigma[V]$ using this procedure results in the number
of recursive calls bounded by $n$ times the number of sets accepted by $\mathcal{R}$. The time bound follows.
%As discussed before, calculate $\sigma[V]$ recursively, using formula (\ref{eqn:reccost}),
%storing all computed values $\sigma[X]$ in memory to avoid recomputation.
%Whenever we access a value $\koszt{\sigma[X]}$ for a set $X$ not accepted by $\mathcal{R}$,
%we take $\koszt{\sigma[X]} = \infty$. As each application of the formula (\ref{eqn:reccost}) gives at most $n$ recursive calls, the bound follows.
\end{proof}

\subsection{The large matching case}

We begin by noticing that the DP algorithm needs to compute $\sigma[X]$ only for those
$X \subseteq V$ that are downward closed, i.e., if $v \in X$ and $u < v$ then $u \in X$.
If there are many constraints in our problem, this alone will suffice to limit the number
of feasible sets considerably, as follows.
Construct an undirected graph $G$ with the vertex set $V$ and edge set $E = \{uv: u < v \vee v < u\}$.
Let $\matching$ be a maximum matching\footnote{Even an inclusion-maximal matching, which can be found greedily, is enough.}
in $G$, which can be found in polynomial time \cite{mucha-sankowski:matching}.
If $X \subseteq V$ is downward closed, and $uv \in \matching$, $u < v$, then it is not possible that $u \notin X$ and $v \in X$.
Obviously checking if a subset is downward closed can be performed in polynomial time, thus we can apply Proposition \ref{prop:cut-dp},
accepting only downward closed subsets of $V$.
This leads to the following lemma:
\begin{lemma}\label{lem:matching}
The number of downward closed subsets of $V$ is bounded by $2^{n-2|\matching|} 3^{|\matching|}$.
If $|\matching| \geq \eps_1 n$, then we can solve the \schedname problem in time
$$T_1(n) = O^\star((3\slash 4)^{\eps_1 n} 2^n).$$\qed
\end{lemma}
Note that for any small positive constant $\eps_1$ the complexity $T_1(n)$ is of required order, i.e., 
$T_1(n) = \Oh(c^n)$ for some $c<2$ that depends on $\eps_1$.
Thus, we only have to deal with the case where $|\matching| < \eps_1 n$.

Let us fix a maximum matching $\matching$, let $\Wvc \subseteq V$ be the set of endpoints of $\matching$,
and let $I_1 = V \setminus \Wvc$.
Note that, as $M$ is a maximum matching in $G$, no two jobs in $I_1$ are bound by a precedence constraint,
and $|\Wvc| \leq 2\eps_1 n$, $|I_1| \geq (1-2\eps_1)n$.
See Figure \ref{fig:matching} for an illustration.

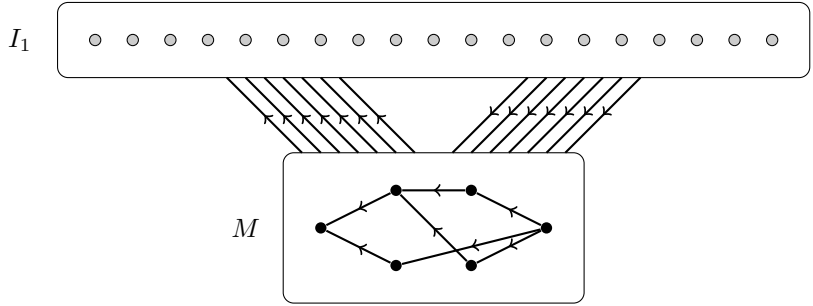
\begin{figure}[htbp]
\begin{center}
\begin{tikzpicture}
    \tikzstyle{wvertex}=[circle,fill=black,minimum size=0.15cm,inner sep=0pt]
    \tikzstyle{ivertex}=[circle,draw=black,fill=light-gray,minimum size=0.15cm,inner sep=0pt]
    \centering
    \begin{scope}[shift={(0,-1.5)}]
      \draw[rounded corners=4pt] (-2, -1) rectangle (2,1) ;
      \draw (-2.5, 0) node {$\Wvc$};
      \node[wvertex] (w1) at (-1.5,0) {};
      \node[wvertex] (w2) at (-0.5,-0.5) {};
      \node[wvertex] (w3) at (-0.5,0.5) {};
      \node[wvertex] (w4) at (0.5,-0.5) {};
      \node[wvertex] (w5) at (0.5,0.5) {};
      \node[wvertex] (w6) at (1.5, 0) {};
    \end{scope}
    \begin{scope}[shift={(0,1)}]
    \draw[rounded corners=4pt] (-5,-0.5) rectangle (5,0.5) ;
    \foreach \x in {1,2,...,19} {
      \node[ivertex] (a\x) at (-5 + \x*0.5, 0) {};
    }
    \draw (-5.5, 0) node {$I_1$};
    \end{scope}
    \begin{scope}[decoration={markings,mark=at position 0.5 with {\arrow{>}}}]
      \foreach \a/\b in {w6/w4,w6/w5,w6/w2,w2/w1,w3/w1,w4/w3,w5/w3} {
        \draw[thick,postaction={decorate}] (\a) -- (\b);
      }
      \foreach \x in {1,2,...,7} {
        \draw[thick,postaction={decorate}] (-2+\x*0.25, -0.5) -- (-3+\x*0.25,0.5);
        \draw[thick,postaction={decorate}] (3-\x*0.25, 0.5) -- (2-\x*0.25,-0.5);
      }
    \end{scope}
\end{tikzpicture}
\caption{An illustration of the case left after Lemma \ref{lem:matching}.
In this and all further figures, an arrow points from the successor job to the predecessor one.}
\label{fig:matching}
\end{center}
\end{figure}

\subsection{High-level overview --- part 2}\label{sec:overview2}

We are left in the situation where there is a small number of ``special'' elements ($\Wvc$),
and the bulk remainder ($I_1$), consisting of elements that are tied by precedence constraints
only to $\Wvc$ and not to each other.

First notice that if $\Wvc$ was empty, the problem would be trivial: with no
precedence constraints we should simply order the tasks from the shortest to the longest.
Now let us consider what would happen if all the constraints between any $u \in I_1$ and
$w \in \Wvc$ would be of the form $u < w$ --- that is, if the jobs from $I_1$ had no predecessors.
For any prefix set candidate $X$ we consider $X_I = X \cap I_1$.
Now for any $x \in X_I$, $y \in I_1 \setminus X_I$ we have an alternative prefix candidate:
the set $X' = (X \cup \{y\}) \setminus \{x\}$. If $\czas{y} < \czas{x}$, there has to be a reason
why $X'$ is not a strictly better prefix candidate than $X$ --- namely, there has to exist
$w \in \Wvc$ such that $x < w$, but $y\not< w$.

A similar reasoning would hold even if not all of $I_1$ had no predecessors, but just some
constant fraction $J$ of $I$ --- again, the only feasible prefix candidates would be those
in which for every $x \in X_I \cap J$ and $y \in J \setminus X_I$ there is a reason (either
$\czas{x} < \czas{y}$ or an element $w \in \Wvc$ which requires $x$, but not $y$) not to exchange them.
It turns out that if $|J| > \eps_2 n$, where $\eps_2 > 2 \eps_1$, this observation suffices
to prove that the number of possible intersections of feasible sets with $J$ is exponentially
smaller than $2^{|J|}$. This is formalized and proved in
Lemma \ref{lem:core}, and is the cornerstone of the whole result.

A typical application of this lemma is as follows: say we have a set $K \subset I_1$ of
cardinality $|K| > 2j$, while
we know for some reason that all the predecessors of elements of $K$ appear on positions
$j$ and earlier. If $K$ is large (a constant fraction of $n$), this is enough to limit
the number of feasible sets to $(2-\eps)^n$. To this end it suffices
to show that there are exponentially
fewer than $2^{|K|}$ possible intersections of a feasible set with $K$. Each such intersection
consists of a set of at most $j$ elements (that will be put on positions $1$ through $j$),
and then a set in which every element has a reason not to be exchanged with something from
outside the set --- and there are relatively few of those by Lemma \ref{lem:core}
--- and when we do the calculations, it turns out the resulting number of
possibilities is exponentially smaller than $2^{|K|}$.

To apply this reasoning, we need to be able to tell that all the prerequisites of a given
element appear at some position or earlier. To achieve this, we need to know the approximate
positions of the elements in $\Wvc$. We achieve this by branching into $4^{|\Wvc|}$ cases,
for each element $w \in \Wvc$ choosing to which of the four quarters of the set
$\{1,\ldots,n\}$ will $\sigma_{opt}(w)$ belong. This incurs a multiplicative cost\footnote{Actually, this bound can be improved to $10^{|\Wvc|/2}$, as $\Wvc$ are endpoints of a matching in the graph corresponding to the set of precedences.}
of $4^{|\Wvc|}$,
which will be offset by the gains from applying Lemma \ref{lem:core}.

We will now repeatedly apply Lemma \ref{lem:core} to obtain information about the
positions of various elements of $I_1$. We will repeatedly say that if ``many'' elements
(by which we always mean more than $\eps n$ for some $\eps$) do not satisfy something,
we can bound the number of feasible sets, and thus finish the algorithm.
For instance, look at those elements of $I_1$ which
can appear in the first quarter, i.e., none of their prerequisites appear in quarters
two, three and four. If there is more than $(\frac{1}{2}+\delta)n$ of them for some constant $\delta > 0$, we can apply the above
reasoning for $j = n\slash 4$ (Lemma \ref{lem:quarters1}).
Subsequent lemmata bound the number of feasible sets if there are many
elements that cannot appear in any of the two first quarters (Lemma \ref{lem:half}),
if {\em less} than $(\frac{1}{2}-\delta)n$ elements can appear in the first
quarter (Lemma \ref{lem:quarters1}) and if a constant fraction of elements
in the second quarter could actually appear in the first quarter
(Lemma \ref{lem:quarters2}). We also apply similar reasoning to elements that
can or cannot appear in the last quarter.

We end up in a situation where we have four groups of elements, each of size
roughly $n \slash 4$, split upon whether they can appear in the first quarter
and whether they can appear in the last one; moreover, those that can appear in
the first quarter will not appear in the second, and those that can appear in
the fourth will not appear in the third. This means that there are two pairs of parts
which do not interact, as the set of places in which they can appear are
disjoint. We use this independence of sorts to construct a different algorithm
than the DP we used so far, which solves our problem in this specific case
in time $O^\star(2^{3n \slash 4 + \eps})$ (Lemma \ref{lem:finish-him}).

As can be gathered from this overview, there are many technical details we
will have to navigate in the algorithm. This is made more precarious by
the need to carefully select all the epsilons. We decided to use symbolic values for
them in the main proof, describing their relationship appropriately, using four constants
$\eps_k$, $k=1,2,3,4$. The constants $\eps_k$ are very small positive reals, and additionally $\eps_k$
is much smaller than $\eps_{k+1}$ for $k=1,2,3$. At each step, we shortly discuss the existence
of such constants. We discuss the choice of optimal values of these constants in Section \ref{sec:values},
although the value we perceive in our algorithm lies rather in the existence
of an $O^\star((2 - \eps)^n)$ algorithm than in the value of $\eps$ (which is admittedly very small).

\subsection{Technical preliminaries}\label{sec:init}
We start with a few simplifications.
First, we add a few dummy jobs with no precedence constraints and zero processing times, so that $n$ is divisible by four.
Second, by slightly perturbing the jobs' processing times, we can assume
that all processing times are pairwise different and, moreover, each ordering has different total completion time.
This can be done, for instance, by replacing time $\czas{v}$ with a pair $(\czas{v}, (n+1)^{\pi(v)-1})$,
where $\pi:V \to \{1,2,\ldots,n\}$ is an arbitrary numbering of $V$. The addition of pairs is performed coordinatewise,
whereas comparison is performed lexicographically.
Note that this in particular implies that the optimal solution is unique, we denote it by $\sigma_{opt}$.
Third, at the cost of an $n^2$ multiplicative overhead,
we guess the jobs $v_{begin} = \sigma_{opt}^{-1}(1)$ and $v_{end}=\sigma_{opt}^{-1}(n)$
and we add precedence constraints $v_{begin} < v < v_{end}$ for each $v \neq v_{begin},v_{end}$.
If $v_{begin}$ or $v_{end}$ were not in $\Wvc$ to begin with, we add them there.

A number of times our algorithm branches into several subcases, in each branch assuming some property
of the optimal solution $\sigma_{opt}$. Formally speaking, in each branch we seek the optimal ordering
among those that satisfy the assumed property. We somewhat abuse the notation and denote by $\sigma_{opt}$
the optimal solution in the currently considered subcase.
Note that $\sigma_{opt}$ is always unique within any subcase, as each ordering
has different total completion time.

For $v \in V$ by $\pred(v)$ we denote the set $\{u \in V : u < v\}$
of predecessors of $v$,
and by $\succc(v)$ we denote the set $\{u \in V : v < u\}$ of successors of $v$.
We extend this notation to subsets of $V$: $\pred(U) = \bigcup_{v \in U} \pred(v)$
and $\succc(U) = \bigcup_{v \in U}\succc(v)$.
Note that for any set $U \subseteq I_1$, both $\pred(U)$ and $\succc(U)$ are subsets of $\Wvc$.

In a few places in this paper we use the following simple bound on binomial coefficients
that can be easily proven using the Stirling's formula.
\begin{lemma}\label{lem:binom}
Let $0 < \alpha < 1$ be a constant. Then
$$\binom{n}{\alpha n} = \Ohstar\left( \left(\frac{1}{\alpha^\alpha (1-\alpha)^{1-\alpha}}\right)^n \right).$$
In particular, if $\alpha \neq 1/2$ then there exists a constant $c_\alpha < 2$
that depends only on $\alpha$ and
$$\binom{n}{\alpha n} = \Ohstar\left(c_\alpha^n\right).$$
\end{lemma}

\subsection{The core lemma}\label{sec:core}

We now formalize the idea of exchanges presented at the beginning of Section \ref{sec:overview2}.
%In the proof of the first case we exchange $u$ with some $v_w$
%whereas in the second case we exchange $v$ with some $u_w$.
\begin{definition}\label{def:xch}
Consider some set $K \subset I_1$, and its subset $L \subset K$.
If there exists $u \in L$ such that for every $w \in \succc(u)$ we can
find $v_w \in (K \cap \pred(w)) \setminus L$ with $t(v_w) < t(u)$ then
we say $L$ is {\em $\succc$-exchangeable} with respect to $K$, otherwise
we say $L$ is {\em non-$\succc$-exchangeable} with respect to $K$.

Similarly, if there exists $v \in (K \setminus L)$ such that for every $w \in \pred(v)$
we can find $u_w \in L \cap \succc(w)$ with $t(u_w) > t(v)$,
we call $L$ {\em $\pred$-exchangeable} with respect to $K$, otherwise we call
it {\em non-$\pred$-exchangeable} with respect to $K$.
\end{definition}

Whenever it is clear from the context, we omit the set $K$ with respect to which its subset is
or is not $\pred$- or $\succc$-exchangeable.

Let us now give some more intuition on the exchangeable sets.
Let $L$ be a non-$\succc$-exchangeable set with respect to $K \subseteq I_1$ and let $u \in L$.
By the definition, there exists $w \in \succc(u)$, such that for all $v_w \in (K \cap \pred(w)) \setminus L$
we have $t(v_w) \geq t(u)$; in other words, all predecessors of $w$ in $K$ that are scheduled after $L$
have larger processing time than $u$ --- which seems like a ``correct'' choice if we are to optimize the total completion time.

On the other hand, let $L = \sigma_{opt}^{-1}(\{1,2,\ldots,i\}) \cap K$ for some $1 \leq i \leq n$ and
assume that $L$ is a $\succc$-exchangeable set with respect to $K$ with a job $u \in L$ witnessing this fact.
Let $w$ be the job in $\succc(u)$ that is scheduled first in the optimal ordering $\sigma_{opt}$.
By the definition, there exists $v_w \in (K \cap \pred(w)) \setminus L$ with $t(v_w) < t(u)$.
It is tempting to decrease the total completion time of $\sigma_{opt}$ by swapping the jobs $v_w$ and $u$ in $\sigma_{opt}$: by the choice of $w$, no precedence constraint
involving $u$ will be violated by such an exchange, so we need to care only about the predecessors of $v_w$.

We formalize the aforementioned applicability of the definition of $\pred$- and $\succc$-exchangeable sets in the following lemma:
\begin{lemma}\label{lem:exchange}
Let $K \subset I_1$.
If for all $v \in K, x \in \pred(K)$ we have that $\sigma_{opt}(v) > \sigma_{opt}(x)$, %insert "that", zeby nie było overline'a
then for any $1 \leq i \leq n$ the set $K \cap \sigma_{opt}^{-1}(\{1, 2,\ldots, i\})$
is non-$\succc$-exchangeable with respect to $K$.

Similarly, if for all $v \in K, x \in \succc(K)$ we have $\sigma_{opt}(v) < \sigma_{opt}(x)$,
then the sets $K \cap \sigma_{opt}^{-1}(\{1, 2,\ldots, i\})$ are non-$\pred$-exchangeable with respect to $K$.
\end{lemma}

\begin{proof}
The proofs for the first and the second case are analogous. However, to help the reader get intuition on exchangeable sets, we provide them both in full detail.
See Figure \ref{rysunek} for an illustration on the $\succc$-exchangeable case.

{\underline{Non-$\succc$-exchangeable sets.}} Assume, by contradiction, that for some $i$ the set $L = K \cap \sigma_{opt}^{-1}(\{1,2,\ldots,i\})$ is $\succc$-exchangeable.
Let $u \in L$ be a job witnessing it. Let $w$ be the successor of $u$ with minimum $\sigma_{opt}(w)$ (there exists one, as $v_{end} \in \succc(u)$). 
By Definition \ref{def:xch}, we have $v_w \in (K \cap \pred(w)) \setminus L$ with $t(v_w) < t(u)$.
As $v_w \in K \setminus L$, we have $\sigma_{opt}(v_w) > \sigma_{opt}(u)$. As $v_w \in \pred(w)$, we have $\sigma_{opt}(v_w) < \sigma_{opt}(w)$.

Consider an ordering $\sigma'$ defined as $\sigma'(u) = \sigma_{opt}(v_w)$, $\sigma'(v_w) = \sigma_{opt}(u)$ and $\sigma'(x) = \sigma_{opt}(x)$ if $x \notin \{u,v_w\}$;
in other words, we swap the positions of $u$ and $v_w$ in the ordering $\sigma_{opt}$. We claim that $\sigma'$ satisfies all the precedence constraints.
As $\sigma_{opt}(u) < \sigma_{opt}(v_w)$, $\sigma'$ may only violates constraints of the form $x < v_w$ and $u < y$. However, if $x < v_w$, then $x \in \pred(K)$
and $\sigma'(v_w) = \sigma_{opt}(u) > \sigma_{opt}(x) = \sigma'(x)$ by the assumptions of the Lemma.
If $u < y$, then $\sigma'(y) = \sigma_{opt}(y) \geq \sigma_{opt}(w) > \sigma_{opt}(v_w) = \sigma'(u)$, by the choice of $w$.
Thus $\sigma'$ is a feasible solution to the considered \schedname{} instance. Since $t(v_w) < t(u)$, we have $\koszt{\sigma'} < \koszt{\sigma_{opt}}$, a contradiction.

{\underline{Non-$\pred$-exchangeable sets.}} Assume, by contradiction, that for some $i$ the set $L = K \cap \sigma_{opt}^{-1}(\{1,2,\ldots,i\})$ is $\pred$-exchangeable.
Let $v \in (K \setminus L)$ be a job witnessing it. Let $w$ be the predecessor of $v$ with maximum $\sigma_{opt}(w)$
(there exists one, as $v_{begin} \in \pred(v)$). 
By Definition \ref{def:xch}, we have $u_w \in L \cap \succc(w)$ with $t(u_w) > t(v)$.
As $u_w \in L$, we have $\sigma_{opt}(u_w) < \sigma_{opt}(v)$. As $u_w \in \succc(w)$, we have $\sigma_{opt}(u_w) > \sigma_{opt}(w)$.

Consider an ordering $\sigma'$ defined as $\sigma'(v) = \sigma_{opt}(u_w)$, $\sigma'(u_w) = \sigma_{opt}(v)$ and $\sigma'(x) = \sigma_{opt}(x)$ if $x \notin \{v,u_w\}$;
in other words, we swap the positions of $v$ and $u_w$ in the ordering $\sigma_{opt}$. We claim that $\sigma'$ satisfies all the precedence constraints.
As $\sigma_{opt}(u_w) < \sigma_{opt}(v)$, $\sigma'$ may only violates constraints of the form $x > u_w$ and $v > y$. However, if $x > u_w$, then $x \in \succc(K)$
and $\sigma'(u_w) = \sigma_{opt}(v) < \sigma_{opt}(x) = \sigma'(x)$ by the assumptions of the Lemma.
If $v > y$, then $\sigma'(y) = \sigma_{opt}(y) \leq \sigma_{opt}(w) < \sigma_{opt}(u_w) = \sigma'(v)$, by the choice of $w$.
Thus $\sigma'$ is a feasible solution to the considered \schedname{} instance. Since $t(u_w) > t(v)$, we have $\koszt{\sigma'} < \koszt{\sigma_{opt}}$, a contradiction.
\end{proof}

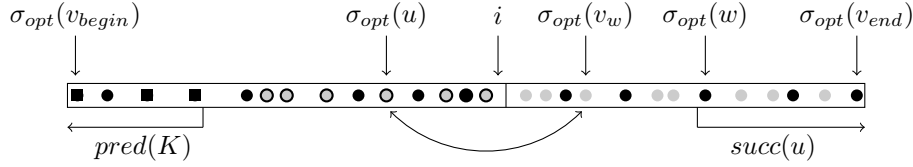
\begin{figure}[htbp]
\begin{center}
\begin{tikzpicture}[scale=0.53]
    \centering
    \draw (-10.0,0.3) rectangle (10.0,-0.3);
    \draw (1.0,0.3) -- (1.0,-0.3);
    \draw (0.8,2.0) node {$i$};
    \path[->] (0.8,1.5) edge (0.8,0.5);
    %\draw (-6.0,0.3) -- (-6.0,-0.3);
    %\draw (-6.2,2.0) node {$j$};
    %\path[->] (-6.2,1.5) edge (-6.2,0.5);
    \foreach \x in {-5,-4.5,-3.5,-2.0,-0.5,0.0,0.5,1.5,2.0,3.0,4.8,5.2,6.9,7.7,9.0}
	{
		\fill[light-gray] (\x,0) circle (0.15);
	}

    \foreach \x in {-9.75,-9,-8,-6.8,-5.5,-2.7,0.0,-1.2,2.5,4.0,6.0,8.2,9.8}
	{
		\fill[black] (\x,0) circle (0.15);
	}
    \foreach \x in {-9.75, -8, -6.8} {
		\fill[black] (\x-0.15, -0.15) rectangle (\x+0.15,0.15);
    }
    \foreach \x in {-5,-4.5,-3.5,-2.0,-0.5,0.0,0.5}
	{
		\draw[thick] (\x,0) circle (0.15);
	}

    \draw (-2.0,2.0) node {$\sigma_{opt}(u)$};
    \path[->] (-2.0,1.5) edge (-2.0,0.5);

    \draw (6.0,2.0) node {$\sigma_{opt}(w)$};
    \path[->] (6.0,1.5) edge (6.0,0.5);

    \draw (5.8,-0.3) -- (5.8,-0.8);
    \path[->] (5.8,-0.8) edge (10,-0.8);
    \draw (7.7, -1.3) node {$\succc(u)$};
    \draw (-6.6,-0.3) -- (-6.6,-0.8);
    \path[->] (-6.6,-0.8) edge (-10,-0.8);
    \draw (-8.1, -1.3) node {$\pred(K)$};

    \draw (3.0,2.0) node {$\sigma_{opt}(v_w)$};
    \path[->] (3.0,1.5) edge (3.0,0.5);

    \draw (-9.8,2.0) node {$\sigma_{opt}(v_{begin})$};
    \path[->] (-9.8,1.5) edge (-9.8,0.5);
    \draw (9.8,2.0) node {$\sigma_{opt}(v_{end})$};
    \path[->] (9.8,1.5) edge (9.8,0.5);

    \path[<->] (-1.9,-0.4) edge [out=315, in=225] (2.9,-0.4);
\end{tikzpicture}
\caption{Figure illustrating the $\succc$-exchangeable case of Lemma \ref{lem:exchange}. Gray circles indicate positions of elements of $K$, black contour indicates that an element is also in $L$. Black squares indicate positions of elements from $\pred(K)$, and black circles --- positions of other elements from $\Wvc$.}
\label{rysunek}
\end{center}
\end{figure}

Lemma \ref{lem:exchange}
means that if we manage to identify a set $K$ satisfying the assumptions
of the lemma, the only sets the DP algorithm has to consider are the
non-exchangeable ones. The following core lemma proves that there are few of
those (provided that $K$ is big enough), and we can identify them easily.

\begin{lemma} \label{lem:core}
For any set $K \subset I_1$ the number of non-$\succc$-exchangeable (non-$\pred$-exchangeable) subsets with regard to $K$ is at most $\sum_{l \leq |\Wvc|} \binom{|K|}{l}$.
Moreover, there exists an algorithm
which checks whether a set is $\succc$-exchangeable ($\pred$-exchangeable) in polynomial
time.
\end{lemma}

The idea of the proof is to construct a function $f$ that encodes each non-exchangeable set
by a subset of $K$ no larger than $\Wvc$.
To show this encoding is injective, we provide a decoding function $g$ and show that
$g\circ f$ is an identity on non-exchangeable sets.

\begin{proof}
As in Lemma \ref{lem:exchange}, the proofs for $\succc$- and $\pred$-exchangeable sets are analogous,
but for the sake or clarity we include both proofs in full detail.

{\underline{Non-$\succc$-exchangeable sets.}} 
For any set $Y \subset K$ we define the function 
$f_Y : \Wvc \to K \cup \{\pinezka\}$ as follows:
for any element $w \in \Wvc$ we define $f_Y(w)$ 
(the {\em least expensive predecessor of} $w$ {\em outside} $Y$) to be the element of
$(K \setminus Y) \cap \pred(w)$ which has the smallest processing time,
or $\pinezka$ if $(K \setminus Y) \cap \pred(w)$ is empty.
We now take $f(Y)$ (the set of the {\em least expensive predecessors outside $Y$})
to be the set $\{f_Y(w) : w \in \Wvc\} \setminus \{\pinezka\}$.
We see that $f(Y)$ is indeed a set of cardinality at most $|\Wvc|$.

Now we aim to prove that
$f$ is injective on the family of non-$\succc$-exchangeable sets.
To this end we define the reverse function $g$. For a set $Z \subset K$
(which we think of as the set of the least expensive predecessors outside some $Y$)
let $g(Z)$ be the set of such elements $v$ of $K$
that there exists $w \in \succc(v)$ such that for any $z_w \in Z \cap \pred(w)$ we have
$t(z_w) > t(v)$. 
Notice, in particular, that $g(Z) \cap Z = \emptyset$, as for $v \in Z$ and $w \in \succc(v)$ we have $v \in Z \cap \pred(w)$.

First we prove $g(f(Y)) \subset Y$ for any $Y \subset K$.
Take any $v \in K \setminus Y$ and consider any $w \in \succc(v)$.
Then $f_Y(w) \neq \pinezka$  and $t(f_Y(w)) \leq t(v)$,
as $v \in (K \setminus Y) \cap \pred(w)$.
Thus $v \notin g(f(Y))$, as for any $w \in \succc(v)$ we can take a witness $z_w = f_Y(w)$
in the definition of $g(f(Y))$.

In the other direction, let us assume that $Y$ does not satisfy $Y \subset g(f(Y))$.
This means we have $u \in Y \setminus g(f(Y))$.
Then we show that $Y$ is $\succc$-exchangeable.
Consider any $w \in \succc(u)$.
As $u \notin g(f(Y))$, by the definition of the function $g$ applied to the set $f(Y)$,
there exists $z_w \in f(Y) \cap \pred(w)$ with
$t(z_w) \leq t(u)$. But $f(Y) \cap Y = \emptyset$, while $u \in Y$;
and as all the values of $t$ are distinct, $t(z_w) < t(u)$
and $z_w$ satisfies the condition for $v_w$ in the definition of
$\succc$-exchangeability.

{\underline{Non-$\pred$-exchangeable sets.}}
For any set $Y \subset K$ we define the function 
$f_Y : \Wvc \to K \cup \{\pinezka\}$ as follows:
for any element $w \in \Wvc$ we define $f_Y(w)$ 
(the {\em most expensive successor of} $w$ {\em in} $Y$) to be the element of
$Y \cap \succc(w)$ which has the largest processing time,
or $\pinezka$ if $Y \cap \succc(w)$ is empty.
We now take $f(Y)$ (the set of the {\em most expensive successors in $Y$})
to be the set $\{f_Y(w) : w \in \Wvc\} \setminus \{\pinezka\}$.
We see that $f(Y)$ is indeed a set of cardinality at most $|\Wvc|$.

Now we aim to prove that $f$ is injective on the family of non-$\pred$-exchangeable sets.
To this end we define the reverse function $g$. For a set $Z \subset K$
(which we think of as the set of most expensive successors in some $Y$)
let $g(Z)$ be the set of such elements $v$ of $K$
that for any $w \in \pred(v)$ there exists a $z_w \in Z \cap \succc(w)$ with
$t(z_w) \geq t(v)$. 
Notice, in particular, that $g(Z) \subset Z$, as for $v \in Z$ the job $z_w = v$ is a good witness for any $w \in \pred(v)$.

First we prove $Y \subset g(f(Y))$ for any $Y \subset K$.
Take any $v \in Y$ and consider any $w \in \pred(v)$.
Then $f_Y(w) \neq \pinezka$  and $t(f_Y(w)) \geq t(v)$,
as $v \in Y \cap \succc(w)$.
Thus $v \in g(f(Y))$, as for any $w \in \pred(v)$ we can take $z_w = f_Y(w)$
in the definition of $g(f(Y))$.

In the other direction, let us assume that $Y$ does not satisfy $g(f(Y)) \subseteq Y$.
This means we have $v \in g(f(Y)) \setminus Y$.
Then we show that $Y$ is $\pred$-exchangeable.
Consider any $w \in \pred(v)$.
As $v \in g(f(Y))$, by the definition of the function $g$ applied to the set $f(Y)$,
 there exists $z_w \in f(Y) \cap \succc(w)$ with
$t(z_w) \geq t(v)$. But $f(Y) \subset Y$, while $v \not\in Y$;
and as all the values of $t$ are distinct, $t(z_w) > t(v)$
and $z_w$ satisfies the condition for $u_w$ in the definition of
$\pred$-exchangeability.

Thus, in both cases, if $Y$ is non-exchangeable then $g(f(Y)) = Y$ (in fact it is possible
to prove in both cases that $Y$ is non-exchangeable iff $g(f(Y))=Y$).
As there are $\sum_{l = 0}^{|\Wvc|} \binom{|K|}{l}$ possible values of $f(Y)$,
the first part of the lemma is proven.
For the second, it suffices to notice that $\succc$- and $\pred$-exchangeability can be checked
in time $\Oh(|K|^2 |\Wvc|)$ directly from the definition. 
\end{proof}

\begin{example}
To illustrate the applicability of Lemma \ref{lem:core}, we analyze the following very simple case:
assume the whole set $\Wvc \setminus \{v_{begin}\}$ succeeds $I_1$, i.e., for every $w\in \Wvc \setminus \{v_{begin}\}$ and $v\in I_1$ we have $w \not< v$.
If $\eps_1$ is small, then we can use the first case of Lemma \ref{lem:exchange} for the whole set $K=I_1$: we have $\pred(K) = \{v_{begin}\}$ and we
only look for orderings that put $v_{begin}$ as the first processed job.
Thus, we can apply Proposition \ref{prop:cut-dp} with algorithm $\mathcal{R}$ that rejects sets $X \subseteq V$ where $X \cap I_1$ is $\succc$-exchangeable with respect to
$I_1$. By Lemma \ref{lem:core}, the number of sets accepted by $\mathcal{R}$ is bounded by $2^{|\Wvc|} \sum_{l \leq |\Wvc|} \binom{|I_1|}{l}$,
which is small if $|\Wvc| \leq \eps_1 n$.
\end{example}

\subsection{Important jobs at $n/2$}\label{sec:half}

As was already mentioned in the overview, the assumptions of Lemma \ref{lem:exchange} are quite strict; therefore, we need to learn a bit more on how $\sigma_{opt}$ behaves on $\Wvc$ in order to distinguish a suitable place for an application. As $|\Wvc| \leq 2\eps_1 n$, we can afford branching into few subcases for every job in $\Wvc$.

Let $A = \{1,2,\ldots,n/4\}$, $B = \{n/4+1, \ldots, n/2\}$, $C = \{n/2+1, \ldots, 3n/4\}$, $D = \{3n/4+1, \ldots, n\}$, i.e., we split $\{1,2,\ldots,n\}$ into quarters.
For each $w \in \Wvc\setminus \{v_{begin},v_{end}\}$ we branch into two cases: whether $\sigma_{opt}(w)$ belongs to $A \cup B$ or $C \cup D$;
however, if some predecessor (successor) of $w$ has been already assigned to $C \cup D$ ($A \cup B$), we do not allow $w$ to be placed in $A \cup B$ ($C \cup D$).
Of course, we already know that $\sigma_{opt}(v_{begin})\in A$ and $\sigma_{opt}(v_{end})\in D$.
Recall that the vertices of $\Wvc$ can be paired into a matching; since for each $w_1 < w_2$, $w_1,w_2 \in \Wvc$ we cannot have $w_1$ placed in $C \cup D$
and $w_2$ placed in $A \cup B$, this branching leads to $3^{|\Wvc|/2} \leq 3^{\eps_1 n}$ subcases, and thus the same overhead in the time complexity.
By the above procedure, in all branches the guesses about alignment of jobs from $\Wvc$ satisfy precedence constraints inside $\Wvc$.

Now consider a fixed branch. 
Let $\Wvc^{AB}$ and $\Wvc^{CD}$ be the sets of elements of $\Wvc$ to be placed in $A \cup B$ and $C \cup D$, respectively.
%For any $\Gamma \in \{A,B,C,D\}$ let $\Wvc^\Gamma$ be the set of elements of $\Wvc$ to be placed in $\Gamma$.
%Moreover let $\Wvc^{AB} = \Wvc^A \cup \Wvc^B$ and $\Wvc^{CD} = \Wvc^C \cup \Wvc^D$.

\begin{figure}[htbp]
\begin{center}
\begin{tikzpicture}
    \centering
%    \begin{scope}[shift={(0,-1.8)}]
%    \foreach \x/\n in {1/A,2/B,3/C,4/D} {
%      \draw (-9 + \x*3, -0.2) rectangle (-6+\x*3,0.2);
%      \draw (-7.5 + \x*3,0) node {$\n$};
%      {\small{\draw (0,-0.4) node {positions};}}
%    }
%    \end{scope}
    \begin{scope}[shift={(0,-1)}]
      \draw[rounded corners=4pt] (-6, -0.3) rectangle (6,0.3) ;
      \draw (0,-0.3) -- (0,0.3);
      \draw (-3, 0) node {$\Wvc^{AB}$};
      \draw (3, 0) node {$\Wvc^{CD}$};
      \coordinate (MAB1) at (-5,0.3);
      \coordinate (MAB2) at (-4,0.3);
      \coordinate (MAB3) at (-3,0.3);
      \coordinate (MAB4) at (-1,0.3);
      \coordinate (MCD1) at (5,0.3);
      \coordinate (MCD2) at (4,0.3);
      \coordinate (MCD3) at (3,0.3);
      \coordinate (MCD4) at (1,0.3);
      \draw (-6.5,0) node {$\Wvc$};
    \end{scope}
    \begin{scope}[shift={(0,2)}]
    \draw[rounded corners=4pt] (-6,-1) rectangle (6,0);
    \draw[rounded corners=4pt] (-6,-1) rectangle (-3,0);
    \draw[rounded corners=4pt] (3,-1) rectangle (6,0);
    \draw[rounded corners=4pt] (-3,-1) rectangle (3,0);
    \coordinate (WAB) at (-4.5,-1);
    \coordinate (WCD) at (4.5,-1);
    \coordinate (I2) at (0,-1);
    \draw (-4.5,-0.5) node {$\Whalf^{AB}$};
    \draw (4.5,-0.5) node {$\Whalf^{CD}$};
    \draw (0,-0.5) node {$I_2$};
    \draw (-6.5,-0.5) node {$I_1$};
    \end{scope}
    \begin{scope}[decoration={markings,mark=at position 0.5 with {\arrow{>}}}]
      \foreach \a/\b in {MAB1/WAB, MAB2/I2, MAB3/WCD, WAB/MAB4, WAB/MCD3, MCD4/WCD, I2/MCD2, WCD/MCD1} {
        \draw[thick,postaction={decorate}] (\b) -- (\a);
      }
    \end{scope}
\end{tikzpicture}
\caption{An illustration of the sets $\Wvc^{AB}$, $\Wvc^{CD}$, $\Whalf^{AB}$ and $\Whalf^{CD}$.}
\label{fig:half}
\end{center}
\end{figure}
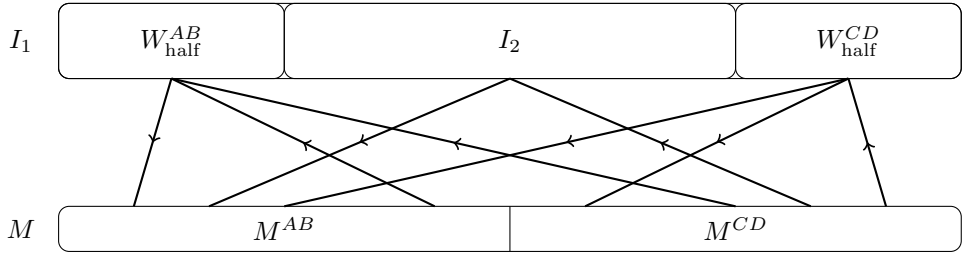

Let us now see what we can learn in a fixed branch about the behaviour of $\sigma_{opt}$ on $I_1$.
Let 
\begin{align*}
\Whalf^{AB} &= \left\{v \in I_1: \exists_w\ \left(w \in \Wvc^{AB} \wedge v < w\right)\right\} \\
\Whalf^{CD} &= \left\{v \in I_1: \exists_w\ \left(w \in \Wvc^{CD} \wedge w < v\right)\right\},
\end{align*}
that is $\Whalf^{AB}$ (resp. $\Whalf^{CD}$) are those elements of $I_1$ which are forced into the first (resp. second) half of $\sigma_{opt}$
by the choices we made about $\Wvc$ (see Figure \ref{fig:half} for an illustration).
If one of the $\Whalf$ sets is much larger than $\Wvc$, we have obtained a gain --- by branching into at most $3^{\eps_1 n}$ branches we gained
additional information about a significant (much larger than $(\log_23) \eps_1 n$) number of other elements (and so we will be able to avoid considering a significant number of
sets in the DP algorithm).
This is formalized in the following lemma:

\begin{lemma}\label{lem:half}
Consider a fixed branch.
If $\Whalf^{AB}$ or $\Whalf^{CD}$ has at least $\eps_2 n$ elements, then the DP algorithm can be augmented to solve the instance in the considered branch in time
$$T_2(n) = \left( 2^{(1-\eps_2)n} + \binom{n}{(1/2 - \eps_2)n} + 2^{\eps_2 n}\binom{(1-\eps_2)n}{n/2} \right) n^{\Oh(1)}.$$
\end{lemma}

\begin{proof}
We describe here only the case $|\Whalf^{AB}| \geq \eps_2 n$. The second case is symmetrical.

Recall that the set $\Whalf^{AB}$ needs to be placed in $A \cup B$ by the optimal ordering $\sigma_{opt}$.
We use Proposition \ref{prop:cut-dp} with an algorithm $\mathcal{R}$ that accepts sets $X \subseteq V$
such that the set $\Whalf^{AB} \setminus X$ (the elements of $\Whalf^{AB}$ not scheduled in $X$) is of size at most $\max(0, n/2 - |X|)$ (the number
of jobs to be scheduled after $X$ in the first half of the jobs).
%\begin{enumerate}
%\item all sets $X$ of size at most $n/2 - \alpha|\Whalf^{AB}|$, there are at most $\binom{n}{(1/2 - \alpha\eps_2)n} n$ such sets;
%\item among sets $X$ of size $n/2 - \alpha|\Whalf^{AB}| \leq |X| \leq n/2$, only those sets for which $|\Whalf^{AB} \setminus X| \leq \alpha|\Whalf^{AB}|$,
%  there are at most $2^{(1-\eps_2)n} \binom{\eps_2 n}{\alpha\eps_2 \cdot n} n$ such sets;
%\item among sets $X$ of size at least $n/2$, only the sets containing $\Whalf^{AB}$, there are at most $2^{(1-\eps_2)n}$ such sets.
%\end{enumerate}
Moreover, the algorithm $\mathcal{R}$ tests if the set $X$ conforms with the guessed sets $\Wvc^{AB}$ and $\Wvc^{CD}$, i.e.:
\begin{align*}
|X| \leq n/2 & \Rightarrow \Wvc^{CD} \cap X = \emptyset \\
|X| \geq n/2 & \Rightarrow \Wvc^{AB} \subseteq X.
\end{align*}
Clearly, for any $1 \leq i \leq n$, the set $\sigma_{opt}^{-1}(\{1,2,\ldots,i\})$ is accepted by $\mathcal{R}$,
as $\sigma_{opt}$ places $\Wvc^{AB} \cup \Whalf^{AB}$ in $A \cup B$ and $\Wvc^{CD}$ in $C \cup D$.

Let us now estimate the number of sets $X$ accepted by $\mathcal{R}$. Any set $X$ of size larger than $n/2$ needs to contain $\Whalf^{AB}$; there are
at most $2^{n-|\Whalf^{AB}|} \leq 2^{(1-\eps_2)n}$ such sets. All sets of size at most $n/2 - |\Whalf^{AB}|$ are accepted by $\mathcal{R}$; there
are at most $n\binom{n}{(1/2 - \eps_2)n}$ such sets. 
Consider now a set $X$ of size $n/2 - \alpha$ for some $0 \leq \alpha \leq |\Whalf^{AB}|$. Such a set needs to contain
$|\Whalf^{AB}|-\beta$ elements of $\Whalf^{AB}$ for some $0 \leq \beta \leq \alpha$ and $n/2 - |\Whalf^{AB}| - (\alpha-\beta)$ elements of $V \setminus \Whalf^{AB}$.
Therefore the number of such sets (for all possible $\alpha$) is bounded by:
\begin{align*}
&\sum_{\alpha=0}^{|\Whalf^{AB}|} \sum_{\beta=0}^\alpha \binom{|\Whalf^{AB}|}{|\Whalf^{AB}|-\beta} \binom{n-|\Whalf^{AB}|}{n/2-|\Whalf^{AB}|-(\alpha-\beta)} \\
&\qquad \leq n^2 \max_{0 \leq \beta \leq \alpha \leq |\Whalf^{AB}|} \binom{|\Whalf^{AB}|}{\beta} \binom{n-|\Whalf^{AB}|}{n/2 + (\alpha-\beta)} \\
&\qquad \leq n^2 2^{|\Whalf^{AB}|} \binom{n-|\Whalf^{AB}|}{n/2} \\
&\qquad \leq n^2 2^{\eps_2 n} \binom{(1-\eps_2)n}{n/2}
\end{align*}
The last inequality follows from the fact that the function $x \mapsto 2^x \binom{n-x}{n/2}$ is decreasing for $x \in [0,n/2]$.
The bound $T_2(n)$ follows.
%
%\begin{align*}
%|X| \leq n/4 & \Rightarrow (\Wvc^B \cup \Wvc^C \cup \Wvc^D) \cap X = \emptyset \\
%n/4 \leq |X| \leq n/2 & \Rightarrow (\Wvc^A \subseteq X \wedge (\Wvc^C \cup \Wvc^D) \cap X = \emptyset) \\
%n/2 \leq |X| \leq 3n/4 & \Rightarrow ((\Wvc^A \cup \Wvc^B) \subseteq X \wedge \Wvc^D \cap X = \emptyset) \\
%3n/4 \leq |X| & \Rightarrow (\Wvc^A \cup \Wvc^B \cup \Wvc^C) \subseteq X.
%\end{align*}

%We now verify that $\sigma_{opt}^{-1}(\{1,2,\ldots,i\})$ is accepted by $\mathcal{R}$ for any $1 \leq i \leq n$. For $i \leq n/2 - \alpha|\Whalf^{AB}|$ it is obvious,
%   and for $i \geq n/2$ it follows from the fact that $\Whalf^{AB}$ is placed in $A \cup B$
%   by $\sigma_{opt}$.
%Finally, for $n/2 - \alpha |\Whalf^{AB}| \leq i \leq n/2$ we have
%$$|\Whalf^{AB} \setminus \sigma_{opt}^{-1}(\{1,2,\ldots,i\})| \leq n/2 - i \leq \alpha|\Whalf^{AB}|.$$
\end{proof}

Note that we have $3^{\eps_1 n}$ overhead so far, due to guessing placement of the jobs from $\Wvc$.
By Lemma \ref{lem:binom}, $\binom{(1-\eps_2) n}{n/2} = \Oh((2-c(\eps_2))^{(1-\eps_2) n})$ and $\binom{n}{(1/2 - \eps_2)n} = \Oh((2-c'(\eps_2))^n)$,
for some positive constants $c(\eps_2)$ and $c'(\eps_2)$ that depend only on $\eps_2$.
Thus, for any small fixed $\eps_2$ we can choose $\eps_1$ sufficiently small
so that $3^{\eps_1 n} T_2(n) = \Oh(c^n)$ for some $c < 2$. Note that $3^{\eps_1 n} T_2(n)$ is an upper bound on the total time spent on processing all the considered subcases.

Let $\Whalf = \Whalf^{AB} \cup \Whalf^{CD}$ and $I_2 = I_1 \setminus \Whalf$. From this point we assume that $|\Whalf^{AB}|, |\Whalf^{CD}| \leq \eps_2 n$, hence $|\Whalf| \leq 2\eps_2 n$
and $|I_2| \geq (1-2\eps_1-2\eps_2)n$.
For each $v \in \Wvc^{AB} \cup \Whalf^{AB}$ we branch into two subcases, whether $\sigma_{opt}(v)$ belongs to $A$ or $B$. Similarly, for each $v \in \Wvc^{CD} \cup \Whalf^{CD}$
we guess whether $\sigma_{opt}(v)$
belongs to $C$ or $D$. Moreover, we terminate branches which are trivially contradicting the constraints.

Let us now estimate the number of subcases created by this branch.
Recall that the vertices of $\Wvc$ can be paired into a matching; since for each $w_1 < w_2$, $w_1,w_2 \in \Wvc$ we cannot have $w_1$ placed in a later segment than $w_2$;
this gives us $10$ options for each pair $w_1 < w_2$.
Thus, in total they are at most $10^{|\Wvc|/2} \leq 10^{\eps_1 n}$ ways of placing vertices of $\Wvc$ into quarters without contradicting the constraints.
Moreover, this step gives us an additional $2^{|\Whalf|} \leq 2^{2\eps_2 n}$ overhead in the time complexity for vertices in $\Whalf$.
Overall, at this point we are considering at most $10^{\eps_1 n} 2^{2\eps_2 n} n^{\Oh(1)}$ subcases.

We denote the set of elements of $\Wvc$ and $\Whalf$ assigned to quarter $\Gamma \in \{A,B,C,D\}$ by $\Wvc^\Gamma$ and $\Whalf^\Gamma$, respectively.

\subsection{Quarters and applications of the core lemma}\label{sec:appl-core}

In this section we try to apply Lemma \ref{lem:core} as follows:
We look which elements of $I_2$ can be placed in $A$ (the set $P^A$) and which cannot (the set $P^{\neg A}$).
Similarly we define the set $P^D$ (can be placed in $D$) and $P^{\neg D}$ (cannot be placed in $D$).
For each of these sets, we try to apply Lemma \ref{lem:core} to some subset of it.
If we fail, then in the next subsection we infer that the solutions in the quarters are partially independent of
each other, and we can solve the problem in time roughly $\Oh(2^{3n/4})$.
Let us now proceed with a more detailed argumentation.

We define the following two partitions of $I_2$:
\begin{align*}
P^{\neg A} &= \left\{v \in I_2: \exists_w \left(w \in \Wvc^B \wedge w < v\right)\right\},\\
P^A &= I_2 \setminus P^{\neg A} = \left\{v \in I_2: \forall_w \left(w < v \Rightarrow w \in \Wvc^A\right)\right\}, \\
P^{\neg D} &= \left\{v \in I_2: \exists_w \left(w \in \Wvc^C \wedge w > v\right)\right\},\\
P^D &= I_2 \setminus P^{\neg D} = \left\{v \in I_2: \forall_w \left(w > v \Rightarrow w \in \Wvc^D\right)\right\}.
\end{align*}
In other words, the elements of $P^{\neg A}$ cannot be placed in $A$ because some of their requirements are in $\Wvc^B$,
   and the elements of $P^{\neg D}$ cannot be placed in $D$ because they are required by some elements of $\Wvc^C$ (see Figure \ref{fig:PA} for an illustration).
 Note that these definitions are independent of $\sigma_{opt}$, so sets $P^\Delta$ for $\Delta\in \{A,\neg A,\neg D, D\}$ can be computed in polynomial time. Let 
\begin{align*}
p^A &= |\sigma_{opt}(P^A) \cap A|,\\
p^B &= |\sigma_{opt}(P^{\neg A}) \cap B|,\\
p^C &= |\sigma_{opt}(P^{\neg D}) \cap C|,\\
p^D &= |\sigma_{opt}(P^D) \cap D|.
\end{align*}

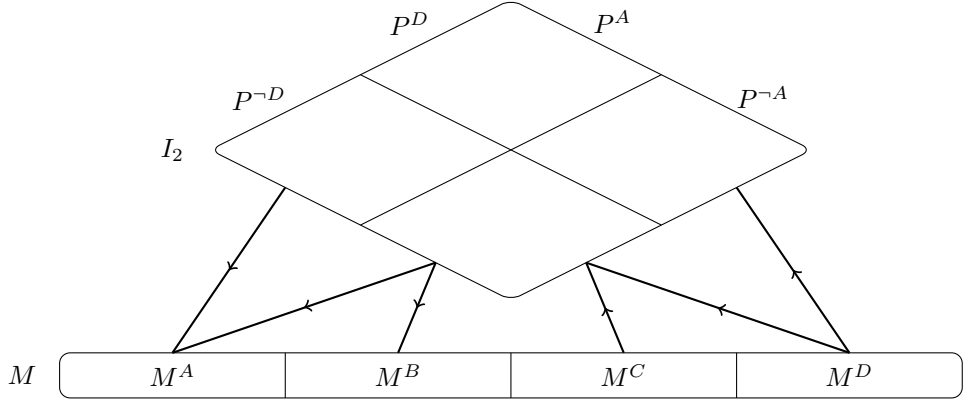
\begin{figure}[htbp]
\begin{center}
\begin{tikzpicture}
    \centering
%    \begin{scope}[shift={(0,-1.8)}]
%    \foreach \x/\n in {1/A,2/B,3/C,4/D} {
%      \draw (-9 + \x*3, -0.2) rectangle (-6+\x*3,0.2);
%      \draw (-7.5 + \x*3,0) node {$\n$};
%      {\small{\draw (0,-0.4) node {positions};}}
%    }
%    \end{scope}
    \begin{scope}[shift={(0,-1)}]
      \draw[rounded corners=4pt] (-6, -0.3) rectangle (6,0.3) ;
      \foreach \x/\n in {1/A,2/B,3/C,4/D} {
        \draw (-7.5 + \x*3, 0) node {$\Wvc^\n$};
        \coordinate (M\n) at (-7.5 + \x*3, 0.3);
      }
      \foreach \x in {1,2,3} {
        \draw (-6+\x*3,-0.3) -- (-6+\x*3,0.3);
      }
      \draw (-6.5,0) node {$\Wvc$};
    \end{scope}
    \begin{scope}[shift={(0,2)},yscale=0.5]
    \draw[rounded corners=4pt] (-2,-2) -- (-4,0) -- (0, 4) -- (4, 0) -- (0, -4) -- (-2,-2);
    \draw (-2,-2) -- (2,2);
    \draw (2,-2) -- (-2,2);
    \coordinate (PA) at (-3,-1);
    \coordinate (PnA) at (-1,-3);
    \coordinate (PD) at (3,-1);
    \coordinate (PnD) at (1,-3);
    \draw (1.35,3.4) node {$P^A$};
    \draw (3.35,1.4) node {$P^{\neg A}$};
    \draw (-1.35,3.35) node {$P^D$};
    \draw (-3.35,1.35) node {$P^{\neg D}$};
    \draw (-4.5,0) node {$I_2$};
    \end{scope}
    \begin{scope}[decoration={markings,mark=at position 0.5 with {\arrow{>}}}]
      \foreach \a/\b in {MA/PA, MA/PnA, MB/PnA, PD/MD, PnD/MD, PnD/MC} {
        \draw[thick,postaction={decorate}] (\b) -- (\a);
      }
    \end{scope}
\end{tikzpicture}
\caption{An illustration of the sets $P^\Delta$ for $\Delta \in \{A,\neg A, \neg D, D\}$ and their relation
  with the sets $\Wvc^\Gamma$ for $\Gamma \in \{A,B,C,D\}$.}
\label{fig:PA}
\end{center}
\end{figure}

Note that $p^\Gamma \leq n/4$ for every $\Gamma\in \{A,B,C,D\}$.
As $p^A = n/4 - |\Wvc^A \cup \Whalf^A|$, $p^D = n/4 - |\Wvc^D \cup \Whalf^D|$, these values can be computed by the algorithm.
We branch into $(1+n/4)^2$ further subcases, guessing the (still unknown) values $p^B$ and $p^C$.

Let us focus on the quarter $A$ and assume that $p^A$ is significantly smaller than $|P^A|/2$ (i.e., $|P^A|/2 - p^a$ is a constant fraction of $n$).
We claim that we can apply Lemma \ref{lem:core} as follows.
While computing $\sigma[X]$, if $|X| \geq n/4$, we can represent $X \cap P^A$ as a disjoint sum of two subsets $X^A_A, X^A_{BCD} \subseteq P^A$. The first one
is of size $p^A$, and represents the elements of $X \cap P^A$ placed in quarter $A$, and the second represents the elements
of $X \cap P^A$ placed in quarters $B \cup C \cup D$. Note that the elements of $X^A_{BCD}$ have all predecessors in the quarter $A$, so by Lemma \ref{lem:exchange} the set $X^A_{BCD}$ has to be non-$\succc$-exchangeable with respect to $P^A \setminus X_A^A$; therefore, by Lemma \ref{lem:core}, we can consider only a very narrow choice of
$X^A_{BCD}$. Thus, the whole part $X\cap P^A$ can be represented by its subset of cardinality at most $p^A$ plus some small information about the rest. If $p^A$ is significantly smaller than $|P^A|/2$, this representation is more concise than simply remembering a subset of $P^A$. Thus we obtain a better bound on the number of feasible sets.

A symmetric situation arises when $p^D$ is significantly smaller than $|P^D|/2$; moreover, we can similarly use Lemma \ref{lem:core} if $p^B$ is significantly smaller than $|P^{\neg A}|/2$ or $p^C$ than $|P^{\neg D}|/2$. This is formalized by the following lemma.
\begin{lemma}\label{lem:quarters0}
If $p^\Gamma < |P^\Delta|/2$ for some $(\Gamma, \Delta) \in \{(A,A),(B,\neg A), (C,\neg D),$ $(D,D)\}$ and $\eps_1 \leq 1/4$,
then the DP algorithm can be augmented to solve the remaining instance in time bounded by
$$T_p (n) = 2^{n-|P^\Delta|} \binom{|P^\Delta|}{p^\Gamma} \binom{n}{|\Wvc|} n^{\Oh(1)}.$$
\end{lemma}

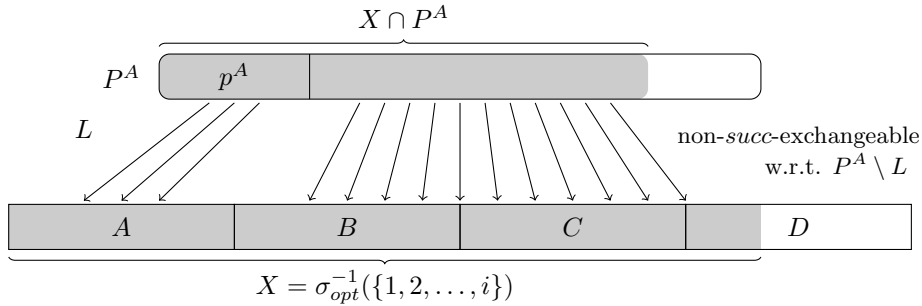
\begin{figure}[htbp]
\begin{center}
\begin{tikzpicture}
    \centering
    \begin{scope}[shift={(0,-2)}]
    \fill[light-gray] (-6,-0.3) rectangle (4,0.3);
    \foreach \x/\n in {1/A,2/B,3/C,4/D} {
      \draw (-9 + \x*3, -0.3) rectangle (-6+\x*3,0.3);
      \draw (-7.5 + \x*3,0) node {$\n$};
    }
    \foreach \x in {1,2,...,5} {
      \coordinate (A\x) at (-6+\x*0.5,0.35);
    }
    \foreach \x in {1,2,...,13} {
      \coordinate (r\x) at (-3 + \x*0.5,0.35);
    }
    \draw[snake=brace] (4,-0.4) -- (-6,-0.4);
    \draw (-1, -0.8) node {$X = \sigma_{opt}^{-1}(\{1,2,\ldots,i\})$};
    \end{scope}

    \begin{scope}[shift={(0,0)}]
      \fill[light-gray,rounded corners=4pt] (-4,-0.3) rectangle (2.5,0.3);
      \draw[snake=brace] (-4,0.4) -- (2.5,0.4);
      \draw (-0.75, 0.8) node {$X \cap P^A$};
      \draw[rounded corners=4pt] (-4, -0.3) rectangle (4,0.3);
      \draw (-4.5, 0) node {$P^A$};
      \draw (-2, -0.3) -- (-2, 0.3);
      \draw (-3, 0) node {$p^A$};
      \foreach \x in {1,2,...,5} {
        \coordinate (pA\x) at (-4 + \x*0.3333,-0.35);
      }
      \foreach \x in {1,2,...,13} {
        \coordinate (PA\x) at (-2 + \x*0.3333,-0.35);
      }
    \end{scope}
    \foreach \x in {2,3,4} {
      \path[->] (pA\x) edge (A\x);
    }
    \foreach \x in {2,3,...,12} {
      \path[->] (PA\x) edge (r\x);
    }
    \draw (-5, -0.7) node {$L$};
    {\small{
       \draw (4.5, -0.8) node {non-$\succc$-exchangeable};
       \draw (5.0,-1.2) node {w.r.t. $P^A \setminus L$};
    }}
\end{tikzpicture}
\caption{An illustration of the proof of Lemma \ref{lem:quarters0} for $(\Gamma,\Delta)=(A,A)$.}
\label{fig:proofPA}
\end{center}
\end{figure}

\begin{proof}
We first describe in detail the case $\Delta=\Gamma=A$, and, later, we shortly describe the other cases that are proven analogously.
An illustration of the proof is depicted on Figure \ref{fig:proofPA}.

On a high-level, we want to proceed as in Proposition \ref{prop:cut-dp}, i.e., use the standard DP algorithm described in Section \ref{sec:high-level1},
while terminating the computation for some unfeasible subsets of $V$.
However, in this case we need to slightly modify the recursive formula used in the computations, and we compute $\sigma[X,L]$ for $X \subseteq V$, $L \subseteq X \cap P^A$.
Intuitively, the set $X$ plays the same role as before, whereas $L$ is the subset of $X \cap P^A$ that was placed in the quarter $A$. Formally, $\sigma[X,L]$ is the ordering of $X$ that attains the minimum total cost among those orderings $\sigma$ for which $L=P^A\cap \sigma^{-1}(A)$.
Thus, in the DP algorithm we use the following recursive formula:
\begin{equation*}
\koszt{\sigma[X,L]} =
\begin{cases}
\min_{v\in \max(X)} \left[\koszt{\sigma[X \setminus \{v\}, L \setminus \{v\}]} + \koszt{v,|X|}\right] \\
    \qquad\qquad\qquad\qquad\qquad\qquad \textrm{if } |X| \leq n/4 \textrm{ and } L=X \cap P^A, \\
+\infty\,\qquad\qquad\qquad\qquad\qquad \textrm{if } |X| \leq n/4 \textrm { and } L \neq X \cap P^A, \\
\min_{v\in \max(X) \setminus L} \left[\koszt{\sigma[X \setminus \{v\}, L]} + \koszt{v,|X|}\right] \\
    \qquad\qquad\qquad\qquad\qquad\qquad \textrm{otherwise.}
\end{cases}
\end{equation*}
In the next paragraphs we describe a polynomial-time algorithm $\mathcal{R}$ that accepts or rejects pairs of subsets $(X,L)$, $X \subseteq V$, $L \subseteq X \cap P^A$;
we terminate the computation on rejected pairs $(X,L)$.
As each single calculation of $\sigma[X,L]$ uses at most $|X|$ recursive calls, the time complexity of the algorithm is bounded by the number of accepted pairs, up to a polynomial multiplicative factor.
We now describe the algorithm $\mathcal{R}$.

First, given a pair $(X,L)$, we ensure that we fulfill the guessed sets $\Wvc^\Gamma$ and $\Whalf^\Gamma$, $\Gamma \in \{A,B,C,D\}$, that is:
%\begin{align*}
%|X| \leq n/4 & \Rightarrow \left(\bigcup_{\Gamma \in \{B,C,D\}} (\Wvc^\Gamma \cup \Whalf^\Gamma)\right) \cap X = \emptyset \\
%n/4 \leq |X| \leq n/2 & \Rightarrow \left(\Wvc^A \cup \Whalf^A \subseteq X \wedge \left(\bigcup_{\Gamma \in \{C,D\}} (\Wvc^\Gamma \cup \Whalf^\Gamma)\right) \cap X = \emptyset\right) \\
%n/2 \leq |X| \leq 3n/4 & \Rightarrow \left(\left(\bigcup_{\Gamma \in \{A,B\}} (\Wvc^\Gamma \cup \Whalf^\Gamma)\right) \subseteq X \wedge (\Wvc^D \cup \Whalf^D) \cap X = \emptyset\right) \\
%3n/4 \leq |X| & \Rightarrow \left(\bigcup_{\Gamma \in \{A,B,C\}} (\Wvc^\Gamma \cup \Whalf^\Gamma)\right) \subseteq X.
%\end{align*}
E.g., we require $\Wvc^B, \Whalf^B \subseteq X$ if $|X| \geq n/2$
and $(\Wvc^B \cup \Whalf^B) \cap X = \emptyset$ if $|X| \leq n/4$. We require similar conditions for other quarters $A$, $C$ and $D$. % (cf. proof of Lemma \ref{lem:half}).
Moreover, we require that $X$ is downward closed. Note that this implies $X \cap P^{\neg A} = \emptyset$ if $|X| \leq n/4$ and $P^{\neg D} \subseteq X$ if $|X| \geq 3n/4$.

Second, we require the following:
\begin{enumerate}
\item If $|X| \leq n/4$, we require that $L = X \cap P^A$ and $|L| \leq p^A$; as $p^A\leq |P^A|/2$, there are at most $2^{n-|P^A|} \binom{|P^A|}{p^A} n$ such pairs $(X,L)$;
\item Otherwise, we require that $|L| = p^A$ and that the set $X \cap (P^A \setminus L)$ is non-$\succc$-exchangeable with respect to $P^A \setminus L$;
by Lemma~\ref{lem:core} there are at most 
$\sum_{l \leq |\Wvc|} \binom{|P^A \setminus L|}{l} \leq n \binom{n}{|\Wvc|}$ (since $|\Wvc| \leq 2\eps_1 n \leq n/2$) non-$\succc$-exchangeable sets with respect to $P^A \setminus L$,
thus there are at most $2^{n-|P^A|} \binom{|P^A|}{p^A} \binom{n}{|\Wvc|} n$ such pairs $(X,L)$.
\end{enumerate}

Let us now check the correctness of the above pruning. Let $0 \leq i \leq n$ and let $X = \sigma_{opt}^{-1}(\{1,2,\ldots,i\})$ and $L = \sigma_{opt}^{-1}(A) \cap X \cap P^A$.
It is easy to see that Lemma \ref{lem:exchange} implies that in case $i \geq n/4$ the set $X \cap (P^A \setminus L)$ is non-$\succc$-exchangeable and
the pair $(X,L)$ is accepted.

Let us now shortly discuss the case $\Gamma = B$ and $\Delta = \neg A$.
Recall that, due to the precedence constraints between $P^{\neg A}$ and $\Wvc^B$,
the jobs from $P^{\neg A}$ cannot be scheduled in the segment $A$. Therefore, while computing $\sigma[X]$ for $|X| \geq n/2$,
we can represent $X \cap P^{\neg A}$ as a disjoint sum of two subsets $X^{\neg A}_B, X^{\neg A}_{CD}$:
the first one, of size $p^B$, to be placed in $B$, and the second one to be placed in $C \cup D$.
Recall that in Section \ref{sec:half} we have ensured that for any $v \in I_2$, all predecessors of $v$ appear
in $\Wvc^{AB}$ and all successors of $v$ appear in $\Wvc^{CD}$.
We infer that all predecessors of jobs in $X^{\neg A}_{CD}$ appear in segments $A$ and $B$
and, by Lemma \ref{lem:exchange}, in the optimal solution the set $X^{\neg A}_{CD}$ is non-$\succc$-exchangeable with respect to $P^{\neg A} \setminus X^{\neg A}_B$,
Therefore we may proceed as in the case of $(\Gamma,\Delta) = (A,A)$; in particular, while computing $\sigma[X,L]$:
\begin{enumerate}
\item If $|X| \leq n/4$, we require that $L = X \cap P^{\neg A} = \emptyset$;
\item If $n/4 < |X| \leq n/2$, we require that $L = X \cap P^{\neg A}$ and $|L| \leq p^B$;
\item Otherwise, we require that $|L| = p^B$ and that the set $X \cap (P^{\neg A} \setminus L)$ is non-$\succc$-exchangeable with respect to $P^{\neg A} \setminus L$.
\end{enumerate}

The cases $(\Gamma,\Delta) \in \{C,\neg D), (D,D)\}$ are symmetrical:
$L$ corresponds to jobs from $P^\Delta$ scheduled to be done in segment $\Gamma$ and
we require that $X \cap (P^\Delta \setminus L)$ is non-$\pred$-exchangeable (instead of non-$\succc$-exchangeable) with respect to $P^\Delta \setminus L$.
The recursive definition of $\koszt{\sigma[X,L]}$ should be also adjusted.
\end{proof}

Observe that if any of the sets $P^\Delta$ for $\Delta \in \{A,\neg A,\neg D, D\}$
is significantly larger than $n/2$ (i.e., larger than $(\frac{1}{2}+\delta)n$ for some $\delta > 0$),
one of the situations in Lemma \ref{lem:quarters0} indeed occurs, since $p^\Gamma \leq n/4$ for $\Gamma \in \{A,B,C,D\}$ and $|\Wvc|$ is small.

\begin{lemma}\label{lem:quarters1}
If $2\eps_1 < 1/4+\eps_3/2$ and at least one of the sets $P^A$, $P^{\neg A}$, $P^{\neg D}$ and $P^{D}$ is of size at least $(1/2+\eps_3)n$, then
the DP algorithm can be augmented to solve the remaining instance in time bounded by
$$T_3(n) = 2^{(1/2-\eps_3)n} \binom{(1/2+\eps_3)n}{n/4} \binom{n}{2\eps_1 n} n^{\Oh(1)}.$$
\end{lemma}
\begin{proof}
The claim is straightforward; note only that the term $2^{n-|P^\Delta|} \binom{|P^\Delta|}{p^\Gamma}$ for $p^\Gamma < |P^\Delta|/2$ is a decreasing function
of $|P^\Delta|$. 
\end{proof}

Note that we have $10^{\eps_1 n} 2^{2\eps_2 n} n^{\Oh(1)}$ overhead so far. As $\binom{(1/2+\eps_3)n}{n/4} = \Oh((2-c(\eps_3))^{(1/2 + \eps_3)n})$ 
for some constant $c(\eps_3) > 0$, for any small fixed $\eps_3$ we can choose sufficiently small $\eps_2$ and $\eps_1$
to have $10^{\eps_1 n} 2^{2\eps_2 n} n^{\Oh(1)} T_3(n) = \Oh(c^n)$ for some $c < 2$.

From this point we assume that $|P^A|, |P^{\neg A}|, |P^{\neg D}|, |P^D| \leq (1/2+\eps_3)n$. As $P^A \cup P^{\neg A} = I_2 = P^{\neg D} \cup P^D$ and $|I_2| \geq (1-2\eps_1-2\eps_2)n$, this implies
that these four sets are of size at least $(1/2-2\eps_1-2\eps_2-\eps_3)n$, i.e., they are of size roughly $n/2$.
Having bounded the sizes of the sets $P^\Delta$ from below, we are able to use Lemma \ref{lem:quarters0} again:
if any of the numbers $p^A$, $p^B$, $p^C$, $p^D$ is significantly smaller than $n/4$ (i.e., smaller than $(\frac{1}{4}-\delta)n$ for some $\delta > 0$),
then it is also significantly smaller than half of the cardinality of the corresponding set $P^\Delta$.

\begin{lemma}\label{lem:quarters2}
Let $\eps_{123} = 2\eps_1 + 2\eps_2 + \eps_3$.
If at least one of the numbers $p^A$, $p^B$, $p^C$ and $p^D$ is smaller than $(1/4-\eps_4)n$ and $\eps_4 > \eps_{123}/2$,
then the DP algorithm can be augmented to solve the remaining instance in time bounded by
$$T_4(n) = 2^{(1/2+\eps_{123})n} \binom{(1/2-\eps_{123})n}{(1/4-\eps_4)n} \binom{n}{2\eps_1 n} n^{\Oh(1)}.$$
\end{lemma}
\begin{proof}
As, before, the claim is a straightforward application of Lemma \ref{lem:quarters0}, and the fact that the term $2^{n-|P^\Delta|} \binom{|P^\Delta|}{p^\Gamma}$ for $p^\Gamma < |P^\Delta|/2$ is a decreasing function of $|P^\Delta|$.
\end{proof}

So far we have $10^{\eps_1 n} 2^{2\eps_2 n} n^{\Oh(1)}$ overhead. Similarly as before, for any small fixed $\eps_4$ if we choose $\eps_1, \eps_2, \eps_3$ sufficiently small, we have $\binom{(1/2-\eps_{123})n}{(1/4-\eps_4)n} = \Oh((2-c(\eps_4))^{(1/2-\eps_{123})n})$
and $10^{\eps_1 n} 2^{2\eps_2 n} n^{\Oh(1)}T_4(n) = \Oh(c^n)$ for some $c<2$. 

Thus we are left with the case when $p^A,p^B,p^C,p^D \geq (1/4-\eps_4)n$.

\subsection{The remaining case}\label{sec:finish}

In this subsection we infer that in the remaining case the quarters $A$, $B$, $C$ and $D$ are somewhat independent, which allows us to develop a faster algorithm. More precisely, note that $p^\Gamma \geq (1/4 - \eps_4)n$, $\Gamma \in \{A,B,C,D\}$, means that
almost all elements that are placed in $A$ by $\sigma_{opt}$ belong to $P^A$, while almost all elements placed in $B$ belong to $P^{\neg A}$. Similarly, almost all elements placed in $D$ belong to $P^D$ and almost all elements placed in $C$ belong to $P^{\neg D}$.
As $P^A \cap P^{\neg A} = \emptyset$ and $P^{\neg D} \cap P^D = \emptyset$, this implies that what happens in the quarters $A$ and $B$, as well as $C$ and $D$, is (almost) independent. This key observation can be used to develop an algorithm that solves this special case in time roughly $\Oh(2^{3n/4})$.

Let $\Wquarter^B = I_2 \cap (\sigma_{opt}^{-1}(B) \setminus P^{\neg A})$ and $\Wquarter^C = I_2 \cap (\sigma_{opt}^{-1}(C) \setminus P^{\neg D})$.
As $p^B, p^C \geq (1/4-\eps_4)n$ we have that $|\Wquarter^B|,|\Wquarter^C| \leq \eps_4 n$.
We branch into at most $n^{2} \binom{n}{\eps_4 n}^2$ subcases, guessing the sets $\Wquarter^B$ and $\Wquarter^C$.
Let $\Wquarter = \Wquarter^B \cup \Wquarter^C$, $I_3 = I_2 \setminus \Wquarter$, $Q^\Delta = P^\Delta \setminus \Wquarter$ for $\Delta \in \{A,\neg A,\neg D,D\}$.
Moreover, let $W^\Gamma = \Wvc^\Gamma \cup \Whalf^\Gamma \cup \Wquarter^\Gamma$ for $\Gamma \in \{A,B,C,D\}$, using the convention $\Wquarter^A = \Wquarter^D = \emptyset$.

Note that in the current branch for any ordering and any $\Gamma \in \{A,B,C,D\}$, the segment $\Gamma$ gets all the jobs from $W^\Gamma$ and $q^\Gamma = n/4- |W^\Gamma|$
jobs from appropriate $Q^\Delta$
($\Delta=A,\neg A,\neg D,D$ for $\Gamma=A,B,C,D$, respectively).
%Note that in the current branch any ordering puts into the segment $\Gamma$ for $\Gamma \in \{A,B,C,D\}$ all the jobs from $W^\Gamma$ and $q^\Gamma = n/4-|W^\Gamma|$ jobs from appropriate $Q^\Delta$
Thus, the behaviour of an ordering $\sigma$ in $A$ influences the behaviour of $\sigma$ in $C$ by the choice of which elements of $Q^A \cap Q^{\neg D}$ are placed
in $A$, and which in $C$. Similar dependencies are between $A$ and $D$, $B$ and $C$, as well as $B$ and $D$ (see Figure \ref{fig:sched-deps}).
In particular, there are no dependencies between $A$ and $B$, as well as $C$ and $D$,
and we can compute the optimal arrangement by keeping track of only three out of four dependencies at once, leading us to an algorithm
running in time roughly $\Oh(2^{3n/4})$.
This is formalized in the following lemma:
\begin{lemma}\label{lem:finish-him}
If $2\eps_1+2\eps_2+\eps_4 < 1/4$ and the assumptions of Lemmata
\ref{lem:matching} and \ref{lem:half}--\ref{lem:quarters2} are not satisfied,
the instance can be solved by an algorithm running in time bounded by
$$T_5(n) = \binom{n}{\eps_4 n}^2 2^{(3/4+\eps_3)n} n^{\Oh(1)}.$$
\end{lemma}

\begin{figure}[htbp]
\begin{center}
\begin{tikzpicture}[scale=0.7]
   \begin{scope}[shift={(-8,0)}]
    \draw (0,0) rectangle (6,6);
    \draw (3,0) -- (3,6);
    \draw (0,3) -- (6,3);
    \draw (1.5,1.5) node {$A$ or $C$};
    \draw (4.5,1.5) node {$B$ or $C$};
    \draw (1.5,4.5) node {$A$ or $D$};
    \draw (4.5,4.5) node {$B$ or $D$};
    \draw (1.5,6.5) node {$Q^A$};
    \draw (4.5,6.5) node {$Q^{\neg A}$};
    \draw (-0.6,4.5) node {$Q^D$};
    \draw (-0.7,1.5) node {$Q^{\neg D}$};
  \end{scope}
  \begin{scope}[shift={(2,1)}]
    \draw (-0.5,-0.5) rectangle (0.5,0.5);
    \draw (3.5,-0.5) rectangle (4.5,0.5);
    \draw (-0.5,3.5) rectangle (0.5,4.5);
    \draw (3.5,3.5) rectangle (4.5,4.5);
    \draw (0,0) node {$D$};
    \draw (4,0) node {$B$};
    \draw (4,4) node {$C$};
    \draw (0,4) node {$A$};
    \draw[<->] (0.7,0) -- (3.3,0);
    \draw[<->] (0.7,4) -- (3.3,4);
    \draw[<->] (0,0.7) -- (0,3.3);
    \draw[<->] (4,0.7) -- (4,3.3);
{\small{
    \draw (2,-0.6) node {$Q^{\neg A} \cap Q^{D}$};
    \draw (2,4.6) node {$Q^{A} \cap Q^{\neg D}$};
    \draw (-1.3, 2) node {$Q^{A} \cap Q^{D}$};
    \draw (5.5, 2) node {$Q^{\neg A} \cap Q^{\neg D}$};
       }}
  \end{scope}
\end{tikzpicture}
\caption{Dependencies between quarters and sets $Q^\Delta$. The left part of the figure illustrates where the jobs from $Q^{\Delta_1} \cap Q^{\Delta_2}$ may be placed.
The right part of the figure illustrates the dependencies between the quarters.}
\label{fig:sched-deps}
\end{center}
\end{figure}
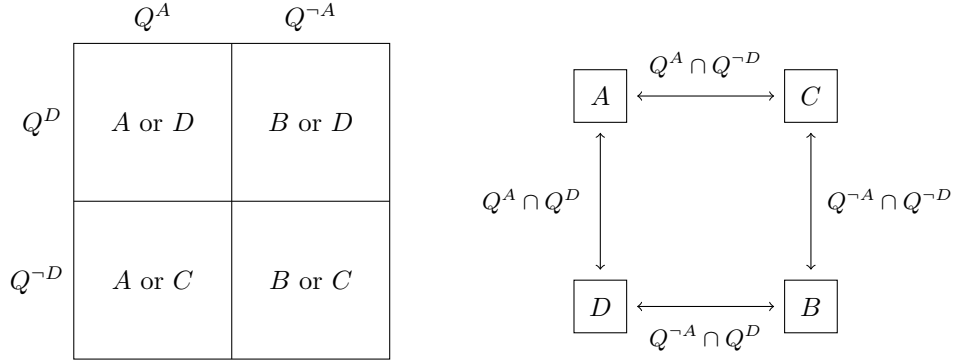

\begin{proof}
Let $(\Gamma,\Delta) \in \{(A,A), (B,\neg A), (C,\neg D), (D,D)\}$.
For each set $Y \subseteq Q^\Delta$ of size $q^\Gamma$, for each bijection (partial ordering) $\sigma^{\Gamma}(Y):Y \cup W^\Gamma \to \Gamma$ let us define its cost as
$$\koszt{\sigma^{\Gamma}(Y)} = \sum_{v \in Y \cup W^\Gamma} \koszt{v,\sigma^{\Gamma}(Y)(v)}.$$
Let $\sigma_{opt}^{\Gamma}(Y)$ be the partial ordering that minimizes the cost (recall that it is unique due to the initial steps in Section \ref{sec:init}).
Note that if we define $Y_{opt}^\Gamma = \sigma_{opt}^{-1}(\Gamma) \cap Q^\Delta$ for $(\Gamma,\Delta) \in \{(A,A), (B,\neg A), (C,\neg D), (D,D)\}$, then
the ordering $\sigma_{opt}$ consists of the partial orderings $\sigma_{opt}^{\Gamma}(Y_{opt}^\Gamma)$.

We first compute the values $\sigma_{opt}^{\Gamma}(Y)$ for all $(\Gamma,\Delta) \in \{(A,A),(B,\neg A),$  $(C,\neg D),(D,D)\}$ and $Y \subseteq Q^\Delta$, $|Y| = q^\Gamma$, by a straightforward
modification of the DP algorithm. For fixed pair $(\Gamma,\Delta)$, the DP algorithm computes $\sigma_{opt}^{\Gamma}(Y)$ for all $Y$ in time
$$2^{|W^\Gamma|+|Q^\Delta|} n^{\Oh(1)} \leq  2^{(2\eps_1+2\eps_2+\eps_4)n + (1/2+\eps_3)n} n^{\Oh(1)} = \Oh(2^{(3/4+\eps_3)n}).$$
The last inequality follows from the assumption $2\eps_1+2\eps_2+\eps_4 < 1/4$.

Let us focus on the sets $Q^A \cap Q^{\neg D}$, $Q^A \cap Q^D$, $Q^{\neg A} \cap Q^{\neg D}$ and $Q^{\neg A} \cap Q^D$. Without loss of generality we assume that
$Q^A \cap Q^{\neg D}$ is the smallest among those. As they all are pairwise disjoint and sum up to $I_2$, we have $|Q^A \cap Q^{\neg D}| \leq n/4$.
We branch into at most $2^{|Q^A \cap Q^{\neg D}|+|Q^{\neg A} \cap Q^D|}$ subcases, guessing the sets
\begin{align*}
Y_{opt}^{AC} &= Y_{opt}^A \cap (Q^A \cap Q^{\neg D}) = (Q^A \cap Q^{\neg D}) \setminus Y_{opt}^C\quad\textrm{and}\\
Y_{opt}^{BD} &= Y_{opt}^B \cap (Q^{\neg A} \cap Q^D) = (Q^{\neg A} \cap Q^D) \setminus Y_{opt}^D.
\end{align*}
Then, we choose the set
$$Y_{opt}^{AD} = Y_{opt}^A \cap (Q^A \cap Q^D) = (Q^A \cap Q^D) \setminus Y_{opt}^D$$
that optimizes
$$\koszt{\sigma_{opt}^{A}(Y_{opt}^{AC} \cup Y_{opt}^{AD})}+\koszt{\sigma_{opt}^{D}(Q^D \setminus (Y_{opt}^{AD} \cup Y_{opt}^{BD})}.$$
Independently, we choose the set
$$Y_{opt}^{BC} = Y_{opt}^B \cap (Q^{\neg A} \cap Q^{\neg D}) = (Q^{\neg A} \cap Q^{\neg D}) \setminus Y_{opt}^C$$
that optimizes
$$\koszt{\sigma_{opt}^{B}(Y_{opt}^{BC} \cup Y_{opt}^{BD})}+\koszt{\sigma_{opt}^{C}(Q^{\neg D} \setminus (Y_{opt}^{BC} \cup Y_{opt}^{AC})}.$$
To see the correctness of the above step, note that $Y_{opt}^A = Y_{opt}^{AC} \cup Y_{opt}^{AD}$, and similarly for other quarters.

The time complexity of the above step is bounded by
\begin{align*}
&2^{|Q^A \cap Q^{\neg D}| + |Q^{\neg A} \cap Q^D|} \left(2^{|Q^A \cap Q^D|} + 2^{|Q^{\neg A} \cap Q^{\neg D}|} \right) n^{\Oh(1)} \\
&\qquad = 2^{|Q^A\cap Q^{\neg D}|} \left(2^{|Q^D|} + 2^{|Q^{\neg A}|}\right) n^{\Oh(1)} \\
&\qquad \leq 2^{(3/4 + \eps_3)n} n^{\Oh(1)}
\end{align*}
and the bound $T_5(n)$ follows.
\end{proof}
So far we have $10^{\eps_1 n} 2^{2\eps_2 n} n^{\Oh(1)}$ overhead. For sufficiently small $\eps_4$ we have $\binom{n}{\eps_4 n} = \Oh(2^{n/16})$ and then for sufficiently small constants $\eps_k$, $k=1,2,3$ we have $10^{\eps_1 n} 2^{2\eps_2 n}n^{\Oh(1)} T_5(n) = \Oh(c^n)$ for some $c < 2$.

\subsection{Numerical values of the constants}\label{sec:values}

\begin{table}[htb]
\begin{center}
\begin{tabular}{|l|l|}
\hline
\textbf{Reference} & \textbf{Running time} \\[2mm]
Lemma \ref{lem:matching} & $T_1(n) = O^\star((3\slash 4)^{\eps_1 n} 2^n)$ \\[2mm]
Lemma \ref{lem:half} & $3^{\eps_1 n} T_2(n) n^{\Oh(1)} = 3^{\eps_1 n} \left( 2^{(1-\eps_2)n} + \binom{n}{(1/2 - \eps_2)n} + 2^{\eps_2 n}\binom{(1-\eps_2)n}{n/2} \right) n^{\Oh(1)}$ \\[2mm]
Lemma \ref{lem:quarters1} & $10^{\eps_1 n} 2^{2\eps_2 n} T_3(n) n^{\Oh(1)} = 10^{\eps_1 n} 2^{2\eps_2 n} 2^{(1/2-\eps_3)n} \binom{(1/2+\eps_3)n}{n/4} \binom{n}{2\eps_1 n} n^{\Oh(1)}$ \\[2mm]
Lemma \ref{lem:quarters2} & $10^{\eps_1 n} 2^{2\eps_2 n} T_4(n) n^{\Oh(1)} = 10^{\eps_1 n} 2^{2\eps_2 n} 2^{(1/2+2\eps_1+2\eps_2+\eps_3)n} \binom{(1/2-2\eps_1-2\eps_2-\eps_3)n}{(1/4-\eps_4)n} \binom{n}{2\eps_1 n} n^{\Oh(1)}$ \\[2mm]
Lemma \ref{lem:finish-him} & $10^{\eps_1 n} 2^{2\eps_2 n} T_5(n) n^{\Oh(1)} = 10^{\eps_1 n} 2^{2\eps_2 n} \binom{n}{\eps_4 n}^2 2^{(3/4+\eps_3)n} n^{\Oh(1)}$ \\[1mm]
        \hline
\end{tabular}
\end{center}
\caption{Summary of running times of all cases of the algorithm.}
\label{table:summary}
\end{table}

Table \ref{table:summary} summarizes the running times of all cases of the algorithm.
Using the following values of the constants:
\begin{align*}
\eps_1 &= 2.677001953125 \cdot 10^{-10} \\
\eps_2 &=  0.00002724628851234912872314453125 \\
\eps_3 &=  0.007010121770270753069780766963958740234375 \\
\eps_4 &=  0.016526753505895047409353537659626454114913940429688 \\
\end{align*}
we get that the running time of our algorithm is bounded by:
$$\Oh\left(\left(2- 10^{-10}\right)^n\right).$$

\section{Conclusion}\label{sec:conc}

We presented an algorithm that solves \schedname{} in $\Oh((2-\eps)^n)$ time for some small $\eps$.
This shows that in some sense \schedname{} appears to be easier than resolving CNF-SAT formulae, which is conjectured to need $2^n$ time (the so-called Strong Exponential Time Hypothesis).
Our algorithm is based on an interesting property of the optimal solution expressed in Lemma \ref{lem:core}, which can be of independent interest.
However, our best efforts to numerically compute an optimal choice of values of the constants $\eps_k$, $k=1,2,3,4$ lead us
to an $\eps$ of the order of $10^{-10}$. Although Lemma \ref{lem:core} seems powerful, we lost a lot while applying it. In particular,
the worst trade-off seems to happen in Section \ref{sec:half}, where $\eps_1$ needs to be chosen much smaller than $\eps_2$.
The natural question is: can the base of the exponent be significantly improved?

\paragraph{Acknowledgements} We thank Dominik Scheder for very useful discussions on the \schedname{} problem during his stay in Warsaw. 
Moreover, we greatly appreciate the detailed comments of anonymous reviewers, especially regarding presentation issues and minor optimizations in our algorithm.

\bibliographystyle{plain}
\bibliography{sched}

\end{document}